\documentclass[11pt]{amsart}
\usepackage{geometry}                % See geometry.pdf to learn the layout options. There are lots.
\usepackage{graphicx}
\usepackage{amssymb}
\usepackage{epstopdf}
\usepackage{graphicx}
\usepackage{amsmath}
\usepackage{mathrsfs}
\usepackage{caption}
\usepackage{subcaption}
\usepackage{color}
\usepackage{enumerate}
\usepackage{yfonts}
\usepackage{mathtools}
\title{Multiscale Invariants of Floquet Topological Insulators}
\author{Guillaume Bal$^\dagger$}
\thanks{$^\dagger$Departments of Statistics and Mathematics and CCAM, University of Chicago, Chicago, IL.\\ {\em Email address:} \texttt{guillaumebal@uchicago.edu}}
\author{Daniel Massatt$^\ddagger$}
\thanks{$^\ddagger$Department of Statistics and CCAM, University of Chicago, Chicago, IL.
\\ {\em Email address:} \texttt{dmassatt@uchicago.edu}}

\numberwithin{equation}{section}
\begin{document}
\maketitle
%\section{}
%\subsection{}

\newcommand{\dm}[1]{{\color{blue} #1}}
\newcommand{\gb}[1]{{\color{magenta} #1}}

% mean integral
\def\Xint#1{\mathchoice
{\XXint\displaystyle\textstyle{#1}}%
{\XXint\textstyle\scriptstyle{#1}}%
{\XXint\scriptstyle\scriptscriptstyle{#1}}%
{\XXint\scriptscriptstyle\scriptscriptstyle{#1}}%
\!\int}
\def\XXint#1#2#3{{\setbox0=\hbox{$#1{#2#3}{\int}$ }
\vcenter{\hbox{$#2#3$ }}\kern-.6\wd0}}
\def\mint{\Xint-}

 \newtheorem{assumption}{Assumption}
\newtheorem{remark}{Remark}
\newtheorem{prop}{Proposition}
\newtheorem{thm}{Theorem}
\newtheorem{lemma}{Lemma}
\newtheorem{definition}{Definition}
\newtheorem{corollary}{Corollary}
\newtheorem{example}{Example}
\numberwithin{definition}{section}
\numberwithin{thm}{section}
\numberwithin{remark}{section}
\numberwithin{prop}{section}
\numberwithin{corollary}{section}
\numberwithin{assumption}{section}
\numberwithin{lemma}{section}

\newcommand{\Imag}{\text{Im}}
\newcommand{\Real}{\text{Re}}
\newcommand{\brak}[1]{ \langle #1 \rangle}
\newcommand{\sgn}{\text{sign}}

\newcommand{\h}{\mathfrak{H}}
\newcommand{\Hone}{\mathcal{H}^1}
\newcommand{\Honeloc}{\mathcal{H}^1_{\text{loc}}}
\newcommand{\wlim}{\text{w-}\hspace{-2mm}\lim}
\newcommand{\C}{\mathcal{C}}
\newcommand{\Tr}{{\rm Tr}}
\newcommand{\F}{\mathcal{F}}
\newcommand{\Hh}{\widehat{H}}
\newcommand{\Ha}{\widehat{H}_k^\text{a}}
\newcommand{\Hn}{\widehat{H}_k^{a,n}}
\newcommand{\Hsmall}{\widehat{H}_k^{aa}}

\newcommand{\A}{\mathcal{A}}
\newcommand{\per}{\text{per}}
\newcommand{\quasi}{\mathcal{S}}
\newcommand{\quasiop}{\mathcal{S}_{\text{op}}}

\newcommand{\Ast}{\text{st}}
\newcommand{\Ath}{\text{th}}
\newcommand{\nmod}{\text{mod}}
\newcommand{\V}{\mathcal{V}}
\newcommand{\Hm}{\mathcal{H}}
\newcommand{\Ln}{\mathcal{L}}
\newcommand{\op}{\text{op}}
\newcommand{\Path}{\mathcal{P}}
\newcommand{\coupling}{\mathfrak{U}}
\newcommand{\gap}{\text{gap}}

\newcommand{\Vm}{\mathcal{V}}

\newcommand{\Hilbert}{\mathcal{H}}
\newcommand{\Bounded}{\mathcal{B}(\Hilbert)}
\newcommand{\HilbertF}{\widehat{\Hilbert}}
\newcommand{\bulkH}{\widehat{H}_{\text{bulk}}}
\newcommand{\Range}{\mathfrak{R}}
\newcommand{\Domain}{\mathcal{D}}
\newcommand{\Hhat}{\widehat{H}}
\newcommand{\xt}{ {\tilde x} }
\newcommand{\yt}{ {\tilde y} }
\newcommand{\supp}{\text{supp}}
\newcommand{\egap}{E_0} % Energy gap in section 2
\newcommand{\x}{\langle x \rangle}
\newcommand{\Wt}{W}                 % Helffer Formula for off-diagonal terms
\newcommand{\Rt}{\mathcal{R}}       % resolvent regularity matrix
\newcommand{\Ihalf}{I_{1/2}}
\newcommand{\Err}{\text{Err}}
\newcommand{\Weyl}{\text{Op}}
\newcommand{\htwo}{\mathfrak{h}}

%GB commands
\newcommand{\pdr}[2]{\dfrac{\partial{#1}}{\partial{#2}}}
\newcommand{\pdrr}[2]{\dfrac{\partial^2{#1}}{\partial{#2}^2}}
\newcommand{\pdrt}[3]{\dfrac{\partial^2{#1}}{\partial{#2}{\partial{#3}}}}
\newcommand{\dr}[2]{\dfrac{d{#1}}{d{#2}}}
\newcommand{\drr}[2]{\dfrac{d^2{#1}}{d{#2}^2}}
\newcommand{\aver}[1]{\langle {#1} \rangle}
\newcommand{\Baver}[1]{\Big\langle {#1} \Big\rangle}
\newcommand{\mH}{{\mathcal H}}
\newcommand{\Ht}{\widehat{H}}
\newcommand{\Nt}{\widetilde{N}}
\newcommand{\Bt}{\widehat{B}}
\newcommand{\Ut}{\widetilde{U}}
\newcommand{\Fourier}{\mathcal{F}}

\newcommand{\emb}{\mathcal{E}}

\DeclarePairedDelimiter\ceil{\lceil}{\rceil}
\DeclarePairedDelimiter\floor{\lfloor}{\rfloor}

\def \Rm {\mathbb R}
\def \Nm {\mathbb N}
\def \Cm {\mathbb C}
\def \Zm {\mathbb Z}
\def \Sm {\mathbb S}
\def \Tm {\mathbb T}
\def \Mm {\mathbb M}
\newcommand{\eps}{\varepsilon}
\newcommand{\E}{\mathbb E}
\newcommand{\dsum}{\displaystyle\sum}
\newcommand{\dint}{\displaystyle\int}
% end GB commands

\begin{abstract}
    This paper analyzes Floquet topological insulators resulting from the time-harmonic irradiation of electromagnetic waves on two dimensional materials such as graphene. We analyze the bulk and edge topologies of approximations to the evolution of the light-matter interaction. Topologically protected interface states are created by spatial modulations of the  drive polarization across an interface. In the high-frequency modulation regime, we obtain a sequence of topologies that apply to different time scales. Bulk-difference invariants are computed in detail and a bulk-interface correspondence is shown to apply. We also analyze a high-frequency high-amplitude modulation resulting in a large-gap effective topology topologically that remains valid only for moderately long times.
\end{abstract}

\section{Introduction}

The field of topological insulators finds important applications in two-dimensional materials as they display transport properties that are in some sense immune to perturbations and imperfections. In particular, conductivity at the interface between two insulators in different topologies takes quantized and non-vanishing values directly related to the topology of the insulators. We refer to, e.g., \cite{Hughes_2013, volovik2003universe, fruchart_2013} as well as their large literature on this well-studied phenomenon. 

The simplest partial differential model allowing us to analyze such a phenomenon is the following Dirac Hamiltonian in two space dimensions
\begin{equation}\label{eq:H2x2}
  H = D\cdot\sigma + m(y) \sigma_3
\end{equation}
with $\sigma=(\sigma_1,\sigma_2)^t$ and $\sigma_{1,2,3}$ the standard Pauli matices and $D=\frac1i(\partial_x,\partial_y)^t$. These operators are used for example in two-band models as an approximation near a gap transition point, or in graphene \cite{Hughes_2013}. The mass term $m(y)$ is a smooth function with prescribed signs as $y\to\pm\infty$ and such that $|m(y)|=m_0>0$ for $|y|\geq y_0>0$, say. The Hamiltonian $H$ is an unbounded self-adjoint operator on $L^2(\Rm^2;\Cm^2)$ for an appropriate domain of definition \cite{bal2}.

The main quantity describing quantized transport along the edge ($y$ close to $0$) is given by the following interface conductivity
\begin{equation}\label{eq:sigmaI}
  \sigma_I = {\rm Tr}\ i[H,P]\varphi'(H).
\end{equation}
The function $0\leq\varphi(h)\leq1$ is defined such that $\varphi(-m_0)=0$, $\varphi(m_0)=1$ and $0\leq\varphi'(h)\in C^\infty_0(-m_0,m_0)$. The term $\varphi'(H)$ thus corresponds to a density of states that are defined within the bulk gap $(-m_0,m_0)$. The spatial function $0\leq P(x)\leq 1$ is a smooth function with $P(x)=0$ for $x\leq x_0-1$ and $P(x)=1$ for $x\geq x_0+1$ for some $x_0\in\Rm$. In the limit $P(x)=H(x-x_0)$, equal to $1$ for $x>x_0$ and $0$ for $x<x_0$, this corresponds to an observable counting the energy density in the interval $x>x_0$. The evolution of this observable is given by $i[H,P]$ in the Heisenberg formalism so that $\sigma_I$ may be interpreted as the rate of charge passing from $x<x_0$ to $x>x_0$ and hence as an interface conductivity. The (operator) trace of that observable against the density of states $\varphi'(H)$ gives the above formula. That the $\sigma_I$ is indeed defined as a trace for $H$ defined in \eqref{eq:H2x2} is justified in \cite{bal2}. Calculations in that reference show that
\begin{equation}\label{eq:valsig}
  2\pi \sigma_I = -\frac12 \big(\sgn(m(y_0))-\sgn(m(-y_0))\big).
\end{equation}
This is a non-negative integer (equal to $\pm1$) when $m(y)$ changes sign across the interface while it vanishes (topologically trivial case) when $m$ has a constant sign at infinity. 

The calculations for the above model are given in \cite{bal2} while \cite{bal3} relates this invariant to a general Fedosov-H\"ormander index that may be computed from the symbol of the above Hamiltonian $H$. This relation also helps us to prove the bulk-interface correspondence and show that the above conductivity may be written as a bulk-difference index constructed from Hamiltonians of the form of $H$ above with $m$ constant \cite{bal3}. 

The main feature of the above index (or the above conductivity) is that it is invariant with respect to a large class of perturbations. For instance, $H_V=H+V$ with $V$ a compactly supported local multiplication operator has the same conductivity as $H$: $\sigma_I[H]=\sigma_I[H+V]$. It is this stability with respect to perturbations and heterogeneities that makes two-dimensional topological insulators potentially useful practically.

\medskip

What drives the above non-trivial topology is the presence of a gap-opening mechanism resulting in $m\not=0$. Indeed, when $m$ is constant, we observe that $H^2=(-\Delta+m^2)I$ so that $H$ has (an absolutely continuous) spectrum given by $\Rm\backslash(-m,m)$. It is difficult to find materials with sufficiently large gaps in practice. One possible method to generate such gaps is to perturb a gapless material by electromagnetic modulation. Starting from a model for graphene ($m=0$ above) and modeling the electromagnetic influence linearly, we obtain the following time-dependent Hamiltonian
\begin{equation}\label{eq:EMH}
  H(t) = (D+A(t)) \cdot\sigma +V(t)
\end{equation}
with $A(t)$ a magnetic potential and $V(t)$ a scalar electric potential.

The above Hamiltonian models the light-matter interaction in the vicinity of a given Dirac point. In the presence of several such points (there are two in standard descriptions of graphene), then a separate analysis needs to be performed at each one of them. The global topological invariants of the problem are then given as the respective sums over all such points. In the case of graphene with a similar influence of the electromagnetic modulation on each Dirac point, the topologies obtained in this paper need to be multiplied by a factor $2$ as in \cite{Perez_2015}. 

We consider $A(t)$ and $V(t)$ time-periodic with large frequency $\Omega$. Note that at any time $t$, $H(t)$ is topologically trivial as no term (in front of $\sigma_3$) is present to open a gap. However, in the regime of fast temporal fluctuations, we expect an effective medium to adequately represent the evolution associated to $H(t)$, hopefully with a gap opening. When the potentials are allowed to depend on $y$, we also expect the corresponding mass term to display sign changes and result in a non-vanishing interface conductivity.

\medskip

The main objective of this paper is to show that this favorable picture holds only in a restricted sense. What we show instead is that approximations to $H(t)$ with different levels of accuracy give rise to different topologies. Moreover, the above topological picture, with a Hamiltonian on $\Cm^2$ with a sign-changing mass term $m(y)$, only appears as an approximation over not-too-large times and for sufficiently smooth initial conditions. 

%In other words, we cannot assign any (useful in that context) topology to $H(t)$. Instead, we show that its evolution is well characterized by different approximations at different time scales that have themselves different topologies. 

Topological insulators involving Hamiltonians with time-periodic coefficients are broadly referred to as Floquet topological insulator (FTI) \cite{fruchart_2016,rudner2019,kitagawa2010}. Floquet theory is a general framework applying to differential equations with periodic coefficient. The terminology of FTI is used in two fairly different situations. A large class of theoretical results exist in what we will call an "adiabatic" regime. Hamiltonians in the Fourier domain are parametrized by a three-dimensional domain $(k_x,k_y,t)$, for instance a three-dimensional torus on which the topological invariants are defined, typically as a three-dimensional winding number \cite{fruchart_2016,rudner2013,sadel2017topological,tauber2018effective}. Such invariants are invariant under rescaling $t\to\lambda t$ and hence the terminology of an adiabatic regime. As we mentioned above, we are interested in the non-adiabatic regime of very rapid temporal oscillations with Hamiltonians $H(t)$ at each $t$ displaying no gap opening and hence no adiabatic (non-trivial) topology. In fact, different levels of approximation give rise to different topologies and we cannot assign any three dimensional topology.  Rather, we try and understand how the evolution of our time-dependent Hamiltonian is approximated by systems that do display non-trivial topologies.

%\medskip

In most of the paper, and following \cite{Perez_2014, Perez_2015}, we analyze a specific model with $A(t)=(\cos(\Omega t),m(y)\sin(\Omega t))$ and $V=0$. In section \ref{sec:hf}, we approximate the evolution of such a Hamiltonian by means of replica models of arbitrary accuracy (over times small compared to $\Omega$ in appropriately rescaled units). These n-replica models take the form of $2(2n+1)\times 2(2n+1)$ systems of equations. A further approximation to the 1-replica model gives rise to a Hamiltonian of the form $H$ above, but only under the assumption that the initial condition is sufficiently smooth. 

The analysis of the topological properties of the n-replica models is given in section \ref{sec:ti}. We use the bulk-interface correspondence derived in \cite{bal3} to relate the interface conductivities to bulk-difference invariants. The long and somewhat intricate computation of the bulk-difference invariants is also presented in detail. 

%\medskip

We conclude the paper in section \ref{sec:at} with a different method of approximation based on an averaging theory. We assume there that the potentials $(A_1(t),V(t,y))$ are highly oscillatory and large while $A_2\equiv0$. The resulting effect of their combination is the opening of a gap of order $O(1)$, unlike the small $O(\Omega^{-1})$ gap observed in the previous sections. We then show that the evolution of the full operator is well approximated by that of the topologically non-trivial effective Hamiltonian over times short compared to $\Omega$ and, as for the preceding model, for initial conditions that are sufficiently regular.

\section{High-frequency approximation theory}
\label{sec:hf}
This section considers a simple model of a high frequency laser driven graphene system. We discuss the Floquet formalism using the replica system, and demonstrate that there is a sequence of simpler Hamiltonians that well describe the electronic structure at varying time scales. We do this by constructing approximations to the evolution operator, and show they well approximate the original evolution operator. 

We will introduce two different approximations for evolution. The first one is based on a Duhamel (Dyson) expansion and generic for all states.
The second approximation only applies to states that are either sufficiently smooth or that live in a small energy window of an approximation's spectrum. With the evolution approximation, we will formally justify the current models used in Section \ref{sec:ti}. %This justification will remain at the formal level as we are not in a normalized setting. %GB: not sure what this means.

%The reason for this latter isn't to improve accuracy of the approximate evolution, but rather to justify a trace formulation over the approximation of observables, which in the next section is used to justify a current model.

Following models considered in the physical literature \cite{lindner2011floquet,Perez_2015}, %{\gb{Remove? There is nothing in that ref. Roderich2013}}, 
we define the laser driven graphene system with an edge as follows:
\begin{align}
\label{e:hamiltonian1}
        &i \partial_t \psi = \widetilde H(t)\psi \\
        \label{e:hamiltonian2}
        & \widetilde H(t) = (D + A(t)) \cdot \sigma \\
        & A(t) = (\cos(\Omega t),m(y)\sin(\Omega t)).
\end{align}
Here, $m(y)$ is a smooth function such that $m(y) = \text{sign}(y)$ for $|y| \geq y_0 > 0$, say. 
We are interested in $\Omega$ large in this paper in the sense that the laser period is small compared to durations of interest. %However, in physics this assumption does not necessarily correspond to high frequency, but could correspond to low laser driving strength when rescaling is taken into account. 
We produce an interface between insulators in two different topologies by assuming the laser is circularly polarized for $y > 0$ and oppositely polarized for $y < 0$ producing a direction of current flow along the edge (in the vicinity of $y=0$).
We find 
\begin{equation*}
A(t) \cdot \sigma = \widetilde B(y) e^{-i\Omega t} + \widetilde B^*(y) e^{i\Omega t}
\end{equation*}
for $\widetilde B(y) = \frac{1}{2}(1+m(y))B_0+\frac{1}{2}(1-m(y))B^*_0$, $B_0 = \begin{pmatrix} 0 & 1 \\ 0 & 0 \end{pmatrix}.$ Then \eqref{e:hamiltonian2} is rewritten as
\begin{equation}
    \widetilde H(t) = D \cdot \sigma + \widetilde B(y)e^{-i\Omega t} + \widetilde B^*(y)e^{i\Omega t}.
\end{equation}
In the high frequency regime $\varepsilon = 1/\Omega\ll1$, we find it convenient to rescale the above problem as follows.
We let
\begin{align} \label{eq:rescaling}
    & \tau=t\Omega=\frac1\eps t,\quad (x,y) \rightarrow \eps(x,y)=\frac1\Omega(x,y),\quad B(y) := \widetilde B(\eps y).
\end{align}
%For /the purposes of this paper, we will find it convenient to rescale the Hamiltonian. We let $\tau = t\Omega$, $(x,y) \rightarrow  (x,y)/\Omega$, and $B(y) = \widetilde B(y/\Omega)$. We will also suggestively write $\varepsilon = 1/\Omega$ as our perturbation parameter. 
We then have the rescaled system:
\begin{align}
    & i \partial_\tau \psi = H(\tau)\psi \\
    & H(\tau) = D\cdot \sigma + \varepsilon\biggl(B(y)e^{-i\tau} + B^*(y)e^{i\tau}\biggr).
\end{align}
Note that $m(y)$ is independent of frequency $\Omega=\eps^{-1}$ so that $B(y)=B_\eps(y)=\tilde B(\eps y)$ depends on $\eps$ in the rescaled Hamiltonian.

 By studying the evolution operator, we show that we can write an approximation for the evolution operator out of a Hamiltonian corresponding to one of these sequential approximations via the Fourier replica model \cite{Perez_2015} that controls the evolution at different time scales. 

This section does not use the edge structure (change of sign of $m(y)$) or even the independence of the coefficients in $x$. All results in the section apply with $B$ a general bounded operator defined over $\Rm^2$ for the replica model approximation in section \ref{sec:rm} and for $B$ with sufficiently smooth coefficients in the approximation by a $2\times2$ system in section \ref{sec:2x2}. 

\subsection{Time evolution}
We start by defining the replica model and the corresponding evolution operator of $H(\tau)$.
We define our Hilbert space of electronic states $\Hilbert := L^2(\mathbb{R}^2;\mathbb{C}^2)$ with the standard $L^2$ norm $\|\cdot\|$. 
The time evolution operator is the operator $U(\tau)$ satisfying
\begin{align}
    & i \partial_\tau U(\tau) = H(\tau) U(\tau)\label{eq:Utau} \\
    & U(0) = I. \nonumber
\end{align}
 
A useful formalism to analyze such problems consists in doubling the number of time variables and introducing \cite{howland,peskin,sambe}
\[
  H_e (t)= H(t)-i\partial_t
\]
with an extended solution given by
\[
  \tilde\psi(x,t,\tau)=e^{-i H_e (t) \tau}\tilde\psi(x,t,0).
\]
The solution to the original problem is then given by $\psi(x,\tau)=\tilde\psi(x,t,\tau)_{|t=\tau}$. We also observe that for $H_e\varphi=\lambda\varphi$, then $\tilde\psi(x,t,\tau)=e^{-i\lambda \tau}\varphi(x,t)$ gives us a solution, $\psi(x,\tau) = e^{-i\lambda \tau}\varphi(x,\tau)$. $\varphi$ is $2\pi$ periodic in time, as can be seen by a simple Fourier transform of the eigenproblem.
We will then take advantage of the periodicity of the driving laser to write a Fourier representation of $H(\tau)$ over the Hilbert space
$$\HilbertF := L^2(\mathbb{R}^2; \mathbb{Z}\otimes\mathbb{C}^2),
$$ 
which we call $\Ht$ \cite{rudner2013}. We define $\Ht$ over $\HilbertF$ by 
\begin{equation}
    [\Ht\hat \varphi]_n = (n + D \cdot \sigma)\varphi_n + \varepsilon B(y)\varphi_{n+1} + \varepsilon B^*(y)\varphi_{n-1},
\end{equation}
where we used $\hat \varphi = (\cdots, \varphi_{1},\varphi_0,\varphi_{-1},\cdots)^t \in \HilbertF$ as a representation of the Fourier modes, $\varphi(x,t) = \sum_\ell \varphi_\ell(x) e^{it\ell}$. We then verify that 
\[
  \Ht_{m,n} = \dfrac{1}{2\pi}\dint_0^{2\pi} e^{-imt}  H_e(t) e^{int}dt
\]
and that
%We chose this Hamiltonian such that
\begin{equation*}
   H_e(t)\varphi(x,t) = \sum_\ell e^{imt} \Ht_{m,m+\ell}\varphi_\ell(x)\qquad \forall m\in\Zm.
\end{equation*}
In other words, $\Ht$ is the Fourier transform of the extended operator $H_e$ restricted to periodic functions. When $H(\tau)$ is independent of time, then the spectrum of $H_e$ or $\Ht$ is nothing but the union of the shifted copies of that of $H$ by any integer. When $\Ht$ is not block diagonal (for us, when $\eps\not=0$), then these shifted copies interact and develop a more complex spectrum.

For $O$ an arbitrary operator over $\HilbertF$, we denote the operator $[O]_{nm} : \Hilbert \rightarrow \Hilbert$ by $[O]_{nm}\psi = [O (e_m \otimes \psi)]_n$ for $[e_m]_n = \delta_{nm}$ for $n,m \in \mathbb{Z}$ (the standard basis). We find the evolution operator is:
\begin{equation*}
    U(\tau) := \sum_k [e^{-i\tau\Ht}]_{0k}.
\end{equation*}
Observe that $[e^{-i\tau\Ht}]_{mk} = e^{i\ell \tau} [e^{-i\tau\Ht}]_{m+\ell,k+\ell}$.
We can then quite easily see this is the evolution operator:
\begin{equation*}
\begin{split}
    i\partial_\tau U(\tau) &= \sum_k [\Ht e^{-i\tau\Ht}]_{0k} 
    = \sum_{k,\ell} \Ht_{0\ell} [e^{-i\tau\Ht}]_{\ell k} 
    = \biggl(\sum_\ell\Ht_{0\ell}e^{-i\ell \tau}\biggr) U(\tau) 
    = H(\tau)U(\tau).
\end{split}
\end{equation*}
This is thus the evolution operator on the domain of $H(\tau)$.

\subsection{Replica Model Approximations}
\label{sec:rm}
Because the Fourier couplings $\varepsilon B(y)$ are small, Fourier modes interact weakly with each other. We thus truncate $\Ht$ in Fourier modes to help us build an approximate evolution operator. If $S \subset \Zm$, we define 
$$\HilbertF_S = L^2(\Rm^2;S \otimes \Cm^2)$$
and the embedding map $\emb_S : \HilbertF_S \rightarrow \HilbertF$ by $[\emb_S\psi]_\ell = \psi_\ell\delta_{\ell \in S}$. Then we define $\Ht_S = \emb_S^*\widehat{H}\emb_S$. In the physics literature, each of the matrix blocks is called a {\em replica}. The Fourier representation $\widehat{H}_S$ is called a {\em replica model} \cite{Perez_2015}. Typically we will be interested in sets of the form
\begin{equation*}
    S_n := \{-n,\cdots,n\},
\end{equation*}
corresponding to the $(2n+1)-$ replica model,
though for the gap at $1/2$ in the quasi-energy spectrum, we would use the following replicas:
\begin{equation*}
    S_n^{(1/2)} := \{-n+1,\cdots,n\}.
\end{equation*}
For brevity, we will denote $\Ht_n = \Ht_{S_n}$. For $n=1$ for example, we have the edge Hamiltonian
\begin{equation}
    \widehat{H}_{1} = \begin{pmatrix} 1 + D \cdot \sigma & \varepsilon B^*(y) & O \\ \varepsilon B(y) & D \cdot \sigma & \varepsilon B^*(y) \\ O & \varepsilon B(y) & -1 + D \cdot \sigma\end{pmatrix}.
\end{equation}
Note that there are $3\times3$ blocks, and so this is called the $3$-replica model.
For $n = 0$, we simply have $\widehat{H}_0 = D\cdot\sigma$, which corresponds to the ungapped bulk graphene Dirac point.  We write the approximate evolution operator:
\begin{equation}
\label{e:evol_n}
     U_n(\tau) = \sum_{k\in S_n} [e^{-i\tau \Ht_{n}}]_{0 k}.
\end{equation}
We wish to emphasize this approximate evolution is for arbitrary states 
$$\psi \in \text{Dom}(H(\tau)).$$
\begin{thm}
\label{thm:evol_bound}
We have for $n \geq 0$,
\begin{equation}
    \|U(\tau)-U_n(\tau)\|_\op \leq \frac{2}{(n+1)!} (2\|B\|_\op\varepsilon \tau)^{n+1}.
\end{equation}
\end{thm}
\begin{proof}
This is a standard application of the Duhamel principle.
We define the linear map (or infinite matrix) $N : \Cm^\Zm \rightarrow \Cm^\Zm$ by $N_{ij} = j\delta_{ij}$, and $I$ be the identity matrix.
We define the coupling matrix  and the block diagonal components as follows:
\begin{align*}
    \Ht' := N \otimes I_2 + I \otimes D\cdot\sigma ,\qquad 
    \Bt := \varepsilon^{-1}(\Ht - \Ht').
\end{align*}
$I_2$ is the $2\times2$ identity matrix.
Let $P_n = \emb_{S_n}\emb_{S_n}^*$ be the projection onto $S_n$ replicas over the space $\HilbertF$. We rewrite our replica approximation slightly:
\begin{equation*}
    \Ht[n] := \Ht' + \varepsilon P_n\Bt P_n.
\end{equation*}
This corresponds to $\Ht_n$, simply extended to leave the diagonal block entries unchanged. Note that this does not modify the evolution operator:
\begin{equation*}
    U_n(\tau) = \sum_{k \in S_n} [e^{-i\tau\Ht[n]}]_{0 k}.
\end{equation*}
%\gb{Check domain} \dm{[The $S_n$? Yes it's correct]}
We define 
\begin{equation*}
    V_n = \Bt - P_{n}\Bt P_{n}.
\end{equation*}
Duhamel's principle then gives us:
\begin{equation*}
    U(\tau) = \sum_k \biggl[e^{-i\tau\Ht[n]} + \frac{\varepsilon}{i}\int_0^\tau e^{-i(\tau-\tau_n)\Ht[n]}V_n e^{-i\tau_n\Ht}d\tau_n\biggr]_{0 k}.
\end{equation*}
We thus have
\begin{equation}
\label{e:recursive1}
    U(\tau)-U_n(\tau) = \sum_k \frac{\varepsilon}{i}\int_0^\tau \biggl[e^{-i(\tau-\tau_n)\Ht[n]}V_n e^{-i\tau_n\Ht}\biggr]_{0 k}d\tau_n.
\end{equation}
Using Duhamel's principle yet again, we notice we have the general formula ($\ell > 0$):
\begin{equation*}
    e^{-i s \Ht[\ell]} = e^{-i s\Ht[\ell-1]} + \frac{\varepsilon}{i}\int_0^s e^{-i(s-s')\Ht[\ell-1]}V_{\ell-1} e^{-is'\Ht}ds'.
\end{equation*}
We will wish to apply this sequentially. To do this, we notice that 
\begin{equation}
\label{e:recursive2}
\begin{split}
    P_0e^{-is\Ht[\ell]}V_\ell &= P_0\biggl(e^{-i s\Ht[\ell-1]} + \frac{\varepsilon}{i}\int_0^s e^{-i(s-s')\Ht[\ell-1]}V_{\ell-1} e^{-is'\Ht[\ell]}ds'\biggr)V_\ell \\
    &=P_0\frac{\varepsilon}{i}\int_0^s e^{-i(s-s')\Ht[\ell-1]}V_{\ell-1} e^{-is'\Ht[\ell]}ds'V_\ell,
\end{split}
\end{equation}
since $P_0e^{-i s\Ht[\ell-1]}V_\ell = 0$. Applying \eqref{e:recursive2} in \eqref{e:recursive1}, we obtain:
\begin{equation}
\begin{split}
    U(\tau)-U_n(\tau) &= (-i\varepsilon)^{n+1}\sum_k \int_0^{\tau}\int_0^{\tau_0}\int_0^{\tau_1}\cdots \int_0^{\tau_{n-1}} \\
    &\hspace{2mm}\times\biggl[e^{-i(\tau-\sum_{j=0}^n\tau_j)\Ht_0} V_0\prod_{\ell = 1}^n \biggl(e^{-i \tau_{\ell-1}\Ht[\ell]}V_\ell\biggr)e^{-i\tau_n\Ht}\biggr]_{0k}d\tau_1\cdots d\tau_n.
\end{split}
\end{equation}
Here we use the (ordered) product notation for matrices $A_\ell$:
$\prod_{\ell = 1}^n A_\ell = A_1A_2\cdots A_n$.
We let
\begin{equation*}
    \Gamma^{(n)} := e^{-i(\tau-\sum_{j=0}^n\tau_j)\Ht_0} V_0\prod_{\ell = 1}^n\biggl(e^{-i \tau_{\ell-1}\Ht[\ell]}V_\ell\biggr).
\end{equation*}
Then $\Gamma_{0 \ell}^{(n)}$ terms only are non-zero for $\ell \in \pm (n+1)$, and hence
\begin{equation}
    \begin{split}
        \sum_k[\Gamma^{(n)}e^{-i\tau_n\Ht}]_{0k} &= \sum_k\sum_{\ell \in \pm (n+1)}[\Gamma^{(n)}]_{0\ell}[e^{-i\tau_n\Ht}]_{\ell k} \nonumber\\
        &= \sum_k\sum_{\ell \in \pm (n+1)}[\Gamma^{(n)}]_{0\ell}[e^{-i\tau_n\Ht}]_{0,k-\ell}e^{-i\ell\tau_n} 
        = \sum_{\ell \in \pm (n+1)}[\Gamma^{(n)}]_{0\ell}e^{-i\ell\tau_n}U(\tau).
    \end{split}
\end{equation}
Now $\|\Gamma^{(n)}\|_\op \leq (2\|B\|_\op)^{n+1}$, and hence
\begin{equation*}
    \|U(\tau)-U_n(\tau)\|_\op \leq 2(2\|B\|_\op \varepsilon)^{n+1} \int_0^\tau\int_0^{\tau_0}\cdots \int_0^{\tau_{n-1}} d\tau_0\cdots d\tau_n = \frac{2}{(n+1)!} (2\|B\|_\op \varepsilon\tau)^{n+1}.
\end{equation*}
The theorem statement follows.
\end{proof}
The replica model $\Ht_n$ may in fact be used to describe the evolution up to a timescale $\tau \sim \varepsilon^{-(n+1)}$ since $H(\tau)$ is periodic. To see this, we let $\tau = 2\pi N + \tau'$ for $0\leq\tau'<2\pi$ and $N \in \Nm$, and define
\begin{equation*}
    U_n'(\tau) = U_n(\tau') [U_n(2\pi)]^N.
\end{equation*}
Note that $U_n$ defined in \eqref{e:evol_n} is not a unitary operator so that $U_n'(\tau)$ is not quite $U_n(\tau)$. 
Likewise, by unitarity of $U(\tau)$ and periodicity of $H(\tau)$,
\begin{equation*}
    U(\tau) = U(\tau') [U(2\pi)]^N.
\end{equation*}
Then we find:
\begin{thm}
Let $c_n = \frac{2(4\pi\|B\|_\op)^{n+1}}{2\pi(n+1)!}$. For the system defined as above,
\begin{equation}
\label{e:error_prop}
    \|U(\tau) - U_n'(\tau)\|_\op \leq c_n(\tau+1) \eps^{n+1} \cdot e^{c_n\tau \eps^{n+1}}.
\end{equation}
\end{thm}
\begin{proof}
We compute 
\[
U(\tau)-U_n'(\tau) = (U(\tau')-U_n(\tau'))U^N(2\pi) + U_n(\tau')(U^N(2\pi)-U^N_n(2\pi)).
\]
The first term is bounded in operator norm by $2\pi c_n\eps^{n+1}$ by the previous theorem and unitarity of $U$. From the above theorem, we find that $\|U_n(2\pi)\|\leq 1+(2\pi c_n)^{n+1}$. The second term is $U_n(\tau')$ times 
\[
    x^N-y^N=\sum_{\ell=1}^{N} x^{N-\ell}(x-y)y^{\ell-1},\qquad \mbox{ for } \quad x=U(2\pi)\ \mbox{ and } \ y=U_n(2\pi).
\]
Thus, still using the previous theorem and the bound on $\|U_n(2\pi)\|$, $U(\tau)-U_n(\tau)$ is bounded in operator norm by $N 2\pi c_n\eps^{n+1}e^{N 2\pi c_n \eps^{n+1}}$, which concludes the proof.
%Let $\delta U = U(2\pi) - U_n(2\pi)$. 
%\[U(\tau) - U_n'(\tau) = \sum_{j=1}^N U(\tau') U(2\pi)^{N-j}(\delta U) U_n(2\pi)^{j-1} + (U(\tau')-U_n(\tau'))\delta U^N.  \]
%From the previous theorem we have
%\[
%\| U(2\pi) - U_n(2\pi)\|_\op \leq \frac{2(4\pi\|B\|_\op)^{n+1}}{(n+1)!}\eps^{n+1}.
%\]
%We instantly obtain using $\|U_n(2\pi)\|_\op \leq 1+2\pi c_n$ that
%\[
%\|U(\tau) - U_n'(\tau)\|_\op \leq 2\pi c_n\eps^{n+1}\sum_{j=0}^N (1+2\pi c_n\eps^{n+1})^j.
%\]
\end{proof}
Both $U_n$ and $U_n'$ are built only out of the $n$-replica model. The above approximation result on $U_n'$ captures the evolution $U(\tau)$ up to the time scale $\tau \sim \varepsilon^{-(n+1)}$. 

%\gb{Not sure what to make of the following paragraph} This decomposition into products of $U_n(1)$ is a crude way of ignoring secular terms, i.e. terms that, while bounded in powers of $\epsilon$, also grow in powers of time. For some methods on removing secular terms, see \cite{ , }. Here we are more interested in studying a gapped replica model $\Ht_n$ to consider topological edge effects, and so the details of secular terms are not of particular interest here.

%
%
% LINK TO SECTION 3
%
%
%subsection{High-frequency evolution} 

%{\gb{Should not be a new section.}}

\medskip

Let us now consider the approximation of observations such as the conductivity in \eqref{eq:sigmaI}. The derivation is entirely formal as it involve traces that may not be defined at this level of generality (see the next section).

Let us construct a functional $g(\Ht_n)$ and consider $\psi \in \Hilbert$ such that 
\begin{equation}
\label{e:replica_to_nonreplica}
\psi = \sum_j [g(\Ht_n)\hat\psi]_j
\end{equation}
for $\hat\psi \in \HilbertF_n$. Since $\Ht_n$ will be shown to have a spectral gap in some cases, $g(\Ht_n)$ plays the role of $\varphi'(H)$ in \eqref{eq:sigmaI}. The evolution of such a wavefield is then given by
\begin{equation*}
    \begin{split}
        U(\tau)\psi &= \sum_k [e^{-i\tau\Ht}]_{0k}\sum_j [g(\Ht_n)\hat\psi]_j   
        = \sum_{\ell,j} e^{i\ell\tau}[e^{-i\tau\Ht}]_{\ell j} [g(\Ht_n)\hat\psi]_j \\&
        = \sum_\ell e^{i\ell\tau}[e^{-i\tau\Ht}\emb_{S_n}g(\Ht_n)\hat\psi]_\ell.
    \end{split}
\end{equation*}
A natural approximation of the evolution is then given in the extended space by
\begin{equation*}
    \Ut_n(\tau)\hat\psi := \sum_\ell e^{i\ell\tau}[e^{-i\tau\Ht_n}g(\Ht_n)\hat\psi]_\ell.
\end{equation*}
By a Duhamel argument similar to those above, we find that
\[
\| U(\tau)\psi - \Ut_n(\tau)\hat\psi\|_\op \lesssim \eps\tau
\]
where $\hat\psi$, $\psi$ are related by \eqref{e:replica_to_nonreplica}. If we considered the density of states  adapted to the low-energy range of $\Ht_n$ and given by
\[
\rho = \sum_{\ell,k}[g(\Ht_n)]_{\ell,k},
\]
then to leading order we have its evolution $\rho(\tau)=U(\tau)\rho U^*(\tau)$ in the Heisenberg formalism approximated by
\[
\rho(\tau) \sim \sum_{\ell,k} [\Ut_n(\tau)g(\Ht_n)\Ut_n^*(\tau)]_{\ell k} = \sum_{\ell,k} [g(\Ht_n)]_{\ell,k}e^{i\tau(\ell-k)}.
\]
Then the time average of the observable $i[H(\tau),P]$ would be presented by
\[
\sigma_I = \mint_0^{2\pi}\Tr\ i[H(\tau),P]\rho(\tau)d\tau = \Tr\  i[\Ht_n,P] g(\Ht_n).
\]
%\gb{Why is there an $\eps$ in the integration?} \dm{[fixed]}

The above results justify replacing the extended operator $\Ht$ by its approximation $\Ht_n$ in the analysis of the evolution operator $U(\tau)$ (justified rigorously) as well as in the evolution of observations of interest such as interface conductivities (justified heuristically). 

It turns out that yet a simpler (effective) $2\times2$ system also provides good accuracy. We now consider its properties.
\subsection{Approximate $2\times2$ system}
\label{sec:2x2}
The replica model with $n=1$ offers an accuracy of order $\eps^2$ on the unitary evolution. However, it involves a $6\times6$ system whose off-diagonal components are small (themselves of order $\eps^2$). In this section, we approximate the $n=1$ replica model by a $2\times2$ system of a form similar to \eqref{eq:H2x2} and with a (small) energy gap.
 
We define an approximate $2\times2$ system by:
\begin{align}\label{eq:htwo}
    &\htwo =  D\cdot\sigma + \varepsilon^2\htwo_1, \qquad 
    \htwo_1 :=  B^*B-BB^*.
\end{align}
We also need to assume $B(y)$ is smooth to obtain a meaningful approximation. In particular, if $\chi$ and $\chi_d$ have supports a distance $d$ apart, then we assume that
\begin{equation}
\label{assump:regularity}
    \|\chi(D\cdot\sigma) B \chi_d(D\cdot\sigma)\|_\op \leq \|B\|_\op e^{-c\varepsilon^{-1} d}.
\end{equation}
To make sense of this estimate, we recall the rescaling \eqref{eq:rescaling}. In the original units, the requirement is that $m(y)$ be sufficiently regular. For instance, we verify that the above relation holds when $m(y) = -1_{(-\infty,0]}\ast \phi(x) + 1_{[0,\infty)}\ast \phi(x)$ for some Gaussian $\phi$. 

We define the $2\times2$ evolution:
\[
U_\htwo(\tau) := e^{-i\tau \htwo}.
\]
Note that $\htwo$ is the effective Hamiltonian obtained (formally) in high-frequency analyses of periodically driven system; see, e.g. \cite{dalibard} or \cite[Section 6]{bukov}.

We no longer expect $U(\tau) - U_\htwo(\tau)$ to be small in operator norm. Mathematically, the unbounded operators $\pm1+D\cdot\sigma$ need to be applied during the elimination procedure and this requires regularity. We first consider the evolution of density observables of the form $g_\alpha(\htwo)$ for $g$ a smooth function supported on $[-c_0,c_0]$ for some $c_0 > 0$ while $g_\alpha(x) = g(x\varepsilon^{-\alpha})$ with $\alpha > 0$. A periodic $\eps$ scaled fluctuation remains on top of this, which we define by
$$
u_\htwo(\tau) = (e^{-i\tau}-1)B^* + (e^{i\tau}-1)B.
$$
We then obtain the following theorem:
%
%
% 2x2 Theorem
%
%
\begin{thm}
\label{thm:2b2}
Assume \eqref{assump:regularity} and let $\htwo$ and $g_\alpha$ be defined as above.  Let $\beta = \alpha$ if $\alpha < 1$, and $\beta < 1$ if $\alpha \geq 1$.
%. If $\alpha > 1$, pick $\beta < \min\{\alpha,1\}$. 
Then if we chose $\eps >0$ sufficiently small, we obtain
\begin{equation*}
    \| \bigl(U(\tau) - U_\htwo(\tau)-\eps u_\htwo(\tau)\bigr)g_\alpha(\htwo)\|_\op \leq C_{\beta} (\tau+1) \eps^{1+\beta},
\end{equation*}
for $\tau  < c_{\beta} \eps^{-(1+\beta)}$. $c_{\beta}$ and $C_{\beta}$ depend on $\|B\|_\op$, $c$, $g$, and $\beta$. $C_{\beta} \rightarrow \infty$ and $c_{\beta} \rightarrow 0$ as $\beta \rightarrow 1$.
\end{thm}
\begin{remark}
Observe that as $\alpha \rightarrow 0$, then $g_\alpha$ has support of order $1$, and thus is coupling to higher order frequencies. The $2\times2$ system breaks down in this case. 
Note also that $(U(\tau)-U_{\htwo}(\tau)-\eps u_\htwo(\tau))\psi$ is small for $\psi$ in the range of $g_\alpha(\htwo)$. This will be generalized in a corollary below.
\end{remark}
\begin{proof}
We can restrict our attention to the case $\tau < 2\pi$. 
consider
\[
\delta := \sup_{0\leq \tau <2\pi}\|[U(\tau)- U_\htwo(\tau)-\eps u_\htwo(\tau)] g_\alpha(\htwo)\|_\op.
\]
It is easy to see $\delta = O(\eps\tau)$ by a Duhamel argument on the first term.
Then we let $\tau = 2\pi N + \tau'$, $0 \leq \tau' < 2\pi$.
\begin{equation*}
\begin{split}
    \|[U(\tau) - U_\htwo(\tau)-\eps u_\htwo(\tau)]g_\alpha(\htwo)\| &\leq \|U(2\pi)^N(U(\tau')-U_\htwo(\tau')-\eps u_\htwo(\tau))g_\alpha(\htwo)\|_\op \\
    & \hspace{5mm}+ \sum_{k =0}^{N-1} \| U(2\pi)^k(U(2\pi)-U_\htwo(2\pi))U_\htwo(\tau-2\pi(1+k))g_\alpha(\htwo)\|_\op \\
    &\leq \delta + \sup_{\tilde \tau}N\| (U(2\pi)-U_\htwo(2\pi)) g_{\alpha,\tilde \tau}(\htwo)\|_\op,
\end{split}
\end{equation*}
where $g_{\alpha,\tau}(x) = e^{-i\tau x}g_\alpha(\htwo)$. It becomes clear from the arguments in the proof that the $e^{-i\tau x}$ factor makes no difference, so we ignore it and focus on bounding for $\tau \leq 2\pi$
\[\|\bigl(U(\tau)-U_\htwo(\tau)-\eps u_\htwo(\tau)\bigr)g_\alpha(\htwo)\|_\op.\]
We choose to use the approximation of the evolution: 
\begin{equation*}
    \widetilde{U}_1(\tau) := \sum_k e^{ik\tau}[e^{-i\tau \Ht_1}]_{k0}.
\end{equation*}
We have
\begin{equation*}
    \| \widetilde{U}_1(\tau)-U(\tau)\|_\op \lesssim  \varepsilon^2
\end{equation*}
as in the previous theorem. It therefore suffices to bound $\|(\widetilde{U}_1(\tau) - U_\htwo(\tau)-\eps u_\htwo(\tau))g_\alpha(\htwo)\|_\op$.
We first show
%
% Lemma 1
%
\begin{lemma}
\label{lemma:first}
For $f \in C_0^3(\Rm)$, $\beta < 2$, and $f_\beta(x) := f(x\eps^{-\beta})$, we have the regularity result
\begin{equation}
\label{e:regularity}
\|f_\beta(\htwo)-f_\beta(D\cdot\sigma)\|_\op \lesssim \eps^{2-\beta}.
\end{equation}
\end{lemma}
\begin{proof}
For this, we briefly recall the tools to apply the Helffer-Sj\"{o}strand formula \cite{davies_1995}; see also \cite{bal2,bal3} for similar contexts. We let
\[
\bar\partial = \frac{1}{2}(\partial_x + i\partial_y).
\]
We define the almost analytic extension with $z=x+iy$
\[
\tilde f(z) = (f(x) + iyf'(x) - \frac{1}{2}y^2f''(x))\lambda(y)
\]
where $\lambda(y)$ is smooth in $y$ supported on $[-2,2]$, and $\lambda(y) = 1$ on $[-1,1]$. Then the Helffer-Sj\"{o}strand formula gives for self-adjoint operator $A$:
\begin{equation*}
    f(A) = -\frac1\pi\int_\Cm \bar\partial \tilde f(z) (z-A)^{-1}dz,
\end{equation*}
with here $dz= dxdy$.
Observe that
\[
\bar\partial \tilde f(z) = -\frac{1}{4}y^2f^{(3)}(x)\lambda(y ) + \frac{1}{2}(f(x) + iyf'(x) -\frac{1}{2} y^2f''(x))\lambda'(y)\ 
%\]
\mbox{ and } \
%\[
\int_\Cm |\bar\partial \tilde f(z)| \cdot |y|^{-2} dz \lesssim 1.
\]
We find
\begin{equation*}
\begin{split}
 f_\beta(\htwo) - f_\beta(D\cdot\sigma) &= \int_\Cm \bar\partial\tilde f(z) \biggl[(z-\eps^{-\beta}\htwo)^{-1}-(z-\eps^{-\beta}D\cdot\sigma)^{-1}\biggr]dz \\
 &= -\eps^{2-\beta}\int_\Cm \bar\partial\tilde f(z) \biggl[(z-\eps^{-\beta}\htwo)^{-1}\htwo_1(z-\eps^{-\beta}D\cdot\sigma)^{-1}\biggr]dz,
\end{split}
\end{equation*}
and hence
\begin{equation*}
\|f_\beta(\htwo)-f_\beta(D\cdot\sigma)\|_\op \lesssim \eps^{2-\beta}.
\end{equation*}
\end{proof}
%
% Lemma 2
%
\begin{lemma}
Let $\beta$ be defined as in the theorem statement. Then
\begin{equation}
    \|U_\htwo(\tau)g_\alpha(\htwo) - [e^{-i\tau\Ht_1}\chi_\beta(\Ht_1)]_{00}g_\alpha(\htwo)\|_\op  \lesssim \eps^{2}.
\end{equation}
% \gb{I'm not sure why we have a constant and $\lesssim$. We need to choose?} \dm{[switched]}
\end{lemma}
\begin{proof}
We begin by taking $\chi_\beta(x) = \chi(x\eps^{-\beta})$ where $\chi$ is smooth, supported on $[-2 c_0,2 c_0]$, and $\chi(x)=1$ on the support of $g$. We further define $\theta$ such that it is supported on $[-3c_0,3c_0]$ and $\theta(x) = 1$ on the support of $\chi$. Likewise we have $\theta_\beta(x) = \theta(x\eps^{-\beta})$. We observe 
\begin{align*}
&g_\alpha(\htwo) = \chi_\beta(\htwo)g_\alpha(\htwo) ,\qquad
\chi_\beta(\htwo) = \theta_\beta(\htwo)\chi_\beta(\htwo).
\end{align*}
We define
\[
\tilde\htwo_1(z) :=  B(z\eps^{\beta}-[1+D\cdot\sigma])^{-1}B^* + B^*(z\eps^\beta+[1-D\cdot\sigma])^{-1}B.
\]
%\dm{We will find this is an approximation of $\htwo_1$.}
%\gb{Is it the same operator as $\htwo$ above? If not we need another notation. It is confusing to know which $\htwo$ is what.} \dm{[DM: added comment and changed notation]}
We define $u_\tau(x) = e^{-i\tau x}$ and $G = \{z :|y|>\eps^2\}$, as we wish to remove the strip $|y| < \eps^2$ so that we can apply the regularity result in Lemma \ref{lemma:first}. We then have
\begin{equation*}
    u_\tau \circ \chi_\beta(\htwo) - [u_\tau \circ \chi_\beta(\Ht_1)]_{00} = \int_{G} u_\tau(z)\bar\partial \tilde \chi(z) \biggl[ (z-\eps^{-\beta}\htwo)^{-1} - [(z-\eps^{-\beta}\Ht_1)^{-1}]_{00}\biggr]dz + O(\eps^{2}).
\end{equation*}
By Schur complements, we obtain
\begin{equation*}
\begin{split}
\int_{G} &u_\tau(z)\bar\partial \tilde \chi(z) \biggl[ (z-\eps^{-\beta}\htwo)^{-1} - [(z-\eps^{-\beta}\Ht_1)^{-1}]_{00}\biggr]dz\\
    &= \int_G u_\tau(z)\bar\partial \tilde \chi(z) \biggl[ (z-\eps^{-\beta}D\cdot\sigma - \eps^{2-\beta}\htwo_1)^{-1} - (z-\eps^{-\beta} D\cdot\sigma - \eps^{2-\beta} \tilde\htwo_1(z))^{-1}\biggr]dz \\
    &= \eps^{2-\beta}\int_G u_\tau(z)\bar\partial \tilde \chi(z) \biggl[ (z-\eps^{-\beta}D\cdot\sigma - \eps^{2-\beta}\tilde\htwo_1(z))^{-1} [\tilde\htwo_1(z) -  \htwo_1] (z-\eps^{-\beta} \htwo)^{-1}\biggr]dz.
\end{split}
\end{equation*}
We thus have:
\begin{equation*}
\begin{split}
([&u_\tau\circ\chi_\beta(\htwo)-[u_\tau\circ\chi_\beta(\Ht_1)]_{00})g_\alpha(\htwo) \\
&=\eps^{2-\beta}\int_G u_\tau(z) \bar\partial \tilde \chi(z) \biggl[ (z-\eps^{-\beta}D\cdot\sigma - \eps^{2-\beta}\tilde\htwo_1(z))^{-1} [\tilde\htwo_1(z) -  \htwo_1] \chi_{\beta}(\htwo)(z-\eps^{-\beta} \htwo)^{-1}\biggr]g_\alpha(\htwo)dz\\
& \hspace{.3cm} + O(\eps^2).
\end{split}
\end{equation*}
Using \eqref{assump:regularity}, we obtain
\begin{equation}
\label{e:zero_entry}
\begin{split}
\|(u_\tau\circ\chi_\beta(\htwo)-[u_\tau\circ\chi_\beta(H_1)]_{00})g_\alpha(\htwo)\|_\op &\lesssim \varepsilon^{2} \\
&\hspace{.3cm}+ \eps^{2-\beta}\sup_{z \in G}\|(\tilde\htwo_1(z)-\htwo_1)\chi_\beta(D\cdot\sigma)\|_\op.
\end{split}
\end{equation}
We have
\begin{equation*}
    \begin{split}
        (\tilde\htwo_1(z) - \htwo_1)\chi_\beta(D\cdot\sigma) &= B\bigl(I+(z \eps^\beta-[1+D\cdot\sigma])^{-1}\bigr)B^*\chi_\beta(D\cdot\sigma) \\
        &+ B^* \bigl(-I + (z \eps^\beta-[-1+D\cdot\sigma])^{-1}\bigr)B\chi_\beta(D\cdot\sigma).
    \end{split}
\end{equation*}
We show the bound on one of these terms as the argument is identical.  Let $d$ be the distance between $\{x:\chi(x) = 1\}$ and $\{x : \theta(x)=1\}^c$, which by definition is non-zero. We let
\begin{equation}
\label{e:omega_bound}
\omega(x) = [1+(z\eps^\beta - 1 - x)^{-1}]\theta_{\beta}(x).
\end{equation}
Considering $x$ bounded away from $-1$, we observe that
$
|\omega(x)| \lesssim \eps^{\beta}.
$
Then by \eqref{assump:regularity}:
\begin{equation*}
\begin{split}
    B\bigl(I+(z \eps^\beta-&[1+D\cdot\sigma])^{-1}\bigr)B^*\chi_\beta(D\cdot\sigma) \\
    &= B\omega(D\cdot\sigma)B^*\chi_\beta(D\cdot\sigma) + O(e^{-cd\eps^{\beta-1}\log(\eps^{-1})}).
\end{split}
\end{equation*}
Therefore,  by \eqref{e:omega_bound},
\[
\|(\tilde\htwo_1(z) - \htwo_1)\chi_\beta(D\cdot\sigma)\|_\op \lesssim \eps^{\beta}.
\]
Putting this bound back into \eqref{e:zero_entry} we obtain
\[
\|(u_\tau \circ \chi_\beta(\htwo)-[u_\tau \circ \chi_\beta(\Ht_1)]_{00})g_\alpha(\htwo)\|_\op \lesssim \eps^{2}.
\]
This concludes our proof.
\end{proof}
We therefore have by the above two lemmas with $\tau \leq 2\pi$:
\begin{equation}
\label{e:2b2_diff}
\begin{split}
(\widetilde{U}_1(\tau)-U_\htwo(\tau))g_\alpha(\htwo) &=\biggl([e^{-i\tau\Ht_1}]_{00}[\chi_\beta(\Ht_1)]_{00}-[u_\tau\circ\chi_\beta(\Ht_1)]_{00}\biggr)g_\alpha(\htwo)  \\
&\hspace{5mm} + \sum_{k \in \pm 1} e^{ik\tau} [e^{-i\tau \Ht_1}]_{k0} \chi_\beta(D\cdot\sigma)g_\alpha(\htwo) + O(\eps^2) \\
= -\sum_{k \in \pm 1}[e^{-i\tau \Ht_1}]_{0k}[\chi_\beta(\Ht_1)]_{k0}g_\alpha(\htwo) 
%\\ & \hspace{5mm} 
&+ \sum_{k \in \pm 1} e^{ik\tau} [e^{-i\tau \Ht_1}]_{k0} \chi_\beta(D\cdot\sigma)g_\alpha(\htwo) + O(\eps^2).
\end{split}
\end{equation}
By a simple Duhamel argument, the second term can be bounded using
\begin{equation*}
    \begin{split}
        [e^{-i\tau \Ht_1}]_{k0}\chi_\beta(D\cdot\sigma) = \biggl[\frac{\eps}{i} \int_0^\tau e^{-i(\tau-s)(D\cdot\sigma+k)}\Bt e^{-i sD\cdot\sigma}ds \biggr]_{k0}\chi_\beta(D\cdot\sigma) + O(\eps^2).
    \end{split}
\end{equation*}
By the same regularity arguments as before, letting $B_1 = B^*$ and $B_{-1} = B$, we obtain
\[
\sum_{k \in \pm 1} e^{ik\tau} [e^{-i\tau H_1}]_{k0} \chi_\beta(D\cdot\sigma) = \eps\sum_{k \in \pm 1} (e^{-i\tau k}-1)\theta_\beta(D\cdot\sigma) B_k \chi_\beta(D\cdot\sigma) + O( \eps^{1+\beta}).
\]
We observe 
$$\eps\theta_\beta(D\cdot\sigma)B_k\chi_\beta(D\cdot\sigma)g_\alpha(\htwo) = \eps B_kg_\alpha(\htwo) + O(\eps^{3-\beta}).$$
Therefore our fluctuation term is given by
\[
u_\htwo(\tau) = \sum_{k \in \pm 1}(1-e^{-i\tau k})B_k.
\]
For the first term on the right-hand side of \eqref{e:2b2_diff}, $[e^{-i\tau H_1}]_{k0} = O(\eps)$ by a simple application of Dumahel's principle. We have by Schur complements
\begin{equation*}
    [\chi_\beta(H_1)]_{k0} = \eps^{2-\beta}\int_\Cm \bar\partial\tilde\chi(z) (z - \eps^{-\beta}(k+D\cdot\sigma))^{-1} B_k (z- \eps^{-\beta}D\cdot\sigma - \eps^{2-\beta}\tilde\htwo_1(z))^{-1}dz.
\end{equation*}
We therefore have
\[
\sum_{k \in \pm 1}[e^{-i\tau \Ht_1}]_{0k}[\chi_\beta(\Ht_1)]_{k0}g_\alpha(\htwo) = O(\eps^{3-\beta}).
\]
Putting all the collected error bounds together, we obtain for $\tau \in [0,2\pi)$ the result:
\[
\| \bigl(U(\tau) - U_\htwo(\tau)-\eps u_\htwo(\tau)\bigr)g_\alpha(\htwo)\|_\op \lesssim  \eps^{\beta+1}.
\]
The theorem result follows.
\end{proof}
Given the rescaling \ref{eq:rescaling}, an original wave function of the form $\psi(\tilde x)$ in the rescaled units will be $\eps\psi(\eps x)$. This motivates the following corollary when we are interested in considering wave packets:
\begin{corollary}
  Suppose we have a wave function $\psi \in C_c^\infty(\Rm^2;\Cm^2)$. Denote
  \[\psi_\alpha(x) = \eps^{\alpha}\psi(x\eps^{\alpha})\]
  for $\alpha > 0$. Then for any $\beta < \min\{\alpha,1\}$ and $s = \frac{2\beta-\alpha}{\alpha-\beta}$, we have
  \begin{equation*}
      \| (U(\tau)-U_\htwo(\tau)-\eps u_\htwo(\tau))\psi_\alpha\|_{L^2} \leq C_{\beta,\alpha}(\tau+1) \eps^{1+\beta}\|\psi\|_{H^s}.
  \end{equation*}
  This holds for $\tau < c_{\beta,\alpha}\eps^{-(1+\beta)}$, where $c_{\beta,\alpha} \rightarrow 0$ and $C_{\beta,\alpha} \rightarrow \infty$ as $\beta \rightarrow 1$. 
\end{corollary}
%\gb{Do we agree that the above wavefunction $\psi_\alpha$ is normalized independently of $\eps$ and $\alpha$?} \dm{[DM: yes, otherwise I agree this wouldn't be a very nice result]}
Note that in dimension $d=2$, $\|\psi_\alpha\|_{L^2(\Rm^2;\Cm^2)} = \|\psi\|_{L^2(\Rm^2;\Cm^2)}$ independently of $\alpha$ and $\eps$.
\begin{proof}
We note that
$
\hat\psi_\alpha(\xi) = \eps^{-\alpha}\hat\psi(\xi \eps^{-\alpha}).
$
Let $\beta < \alpha$ and consider $\chi$ supported on $[-c_0,c_0]$ for some $c_0 > 0$ and $\chi(x) = 1$ on $\frac{1}{2}[-c_0,c_0]$, $\chi \geq 0$. Then we have a partition of unity using $\chi_\beta$ and $1-\chi_\beta$. We observe
\begin{equation*}
\begin{split}
    \| (1-\chi_\beta)(D\cdot\sigma)\psi_\alpha\|_{L_2} & \leq \biggl(\int_{|\xi| > \eps^{\beta}c_0/2} |\hat \psi_\alpha(\xi)|^2d\xi\biggr)^{1/2} \\
    & \leq \eps^{\alpha-\beta} \biggl(\int_{|\xi| > \eps^{(\beta-\alpha)}c_0/2} |\hat\psi(\xi)|^2d\xi\biggr)^{1/2} 
    \lesssim \eps^{( \alpha-\beta)(1+s)} \|\psi\|_{H^s}.
\end{split}
\end{equation*}
Using Lemma \ref{lemma:first} and a basic Duhamel argument, we have
\begin{equation*}
    \| (U(\tau) - U_\htwo(\tau)-\eps u_\htwo(\tau))(\chi_\beta(D\cdot\sigma)-\chi_\beta(\htwo))\|_\op \lesssim \eps^{3-\beta}.
\end{equation*}
If we write
\begin{equation*}
     (U(\tau) - U_\htwo(\tau)-\eps u_\htwo(\tau))\psi_\alpha = \biggl(U(\tau) - U_\htwo(\tau)-\eps u_\htwo(\tau)\biggr)\cdot\biggl((1-\chi_\beta)(D\cdot\sigma) + \chi_\beta(\htwo) + [\chi_\beta(D\cdot\sigma) - \chi_\beta(\htwo)]\biggr)\psi_\alpha,
\end{equation*}
the theorem statement instantly follows from the above estimates and Theorem \ref{thm:2b2}.
\end{proof}
As before, we use a heuristic argument based on the evolution bounds to justify the current formula for densities corresponding to $\htwo$.
Let 
$$\rho = g_\alpha(\htwo).$$
Then we have up to $O(\eps^{1+\beta})$
\begin{equation*}
    \begin{split}
        \rho(\tau) &= (U_\htwo(\tau) + \eps u_\htwo(\tau))\rho(U_\htwo^*(\tau) + \eps u_\htwo^*(\tau)) \\
        &= \rho + \eps (u_\htwo(\tau)\rho + \rho u_\htwo^*(\tau)) + O(\eps^2).
    \end{split}
\end{equation*}
Time averaged over a period of the driving force, we then obtain the current
\[
\sigma_I = \Tr\ i[\htwo,P]g_\alpha(\htwo).
\]
This corresponds to a current of the Dirac model with mass at an interface, which is analyzed in \cite{bal2}. This shall be considered more in the next section.

\section{Replica Topologies}
\label{sec:ti}
The evolution operator $U(\tau)$ is well-approximated for large but not-too-large times either by $(2n+1)-$replica models or by the central $2\times2$ system considered in the preceding section. In this section, we show that all these levels of approximations involve Hamiltonians with precise topological invariants of the form of bulk-difference invariants or interface conductivities. Moreover, the invariants strongly depend on $n$ with values that diverge as $n\to\infty$. This gives an example of different levels of approximation of the unitary $U(\tau)$ displaying different values of the topological invariant. 

Following \cite{bal3}, we first compute the bulk-difference invariants of the different approximations in section \ref{sec:bulk} and then show that the bulk-interface correspondence in \cite{bal3} applies to such approximations in section \ref{sec:cond}.
\subsection{Bulk calculations}
\label{sec:bulk}
We first focus on calculating the bulk invariants corresponding to $\Ht_n$ from the previous section. We define the bulk infinite matrix for the Bloch wave $\xi$:
\begin{equation*}
    (\Ht \psi)_n = (n + \xi\sigma)\psi_n + \varepsilon (B \psi_{n-1} + B^*\psi_{n+1}).
\end{equation*}
Here, $\xi\sigma\equiv\xi\cdot\sigma=\xi_1\sigma_1+\xi_2\sigma_2$ and $\Ht = \Ht(\xi)$.

The truncated Hamiltonians are simply projecting $\Ht$ onto $S_n$ replicas, i.e. $\Ht_n = \emb_n^*\Ht\emb_n$, where here we define the embedding $\emb_n : \ell^2(S_n \otimes \Cm^2) \rightarrow \ell^2(\Zm\otimes \Cm^2)$ in parallel to before by
\begin{equation*}
    (\emb_n\psi)_k = \delta_{|k| \leq n} \psi_k,
\end{equation*}
where $\psi = (\psi_{n},\cdots,\psi_{-n})$, $\psi_j \in \Cm^2$.
For $n=1$, we have the $3=(2n+1)-$replica $6\times6$ model
\[
  \left( \begin{matrix}
    1+\xi\sigma & \eps B^* & 0 \\ \eps B & \xi\sigma & \eps B^* \\ 0 & \eps B & -1+\xi\sigma 
\end{matrix} \right) \psi = E\psi.
\]
We consider
%\[
%  B=\left(\begin{matrix} 0 & \beta \\ \alpha & 0 \end{matrix}\right),\quad  B^*=\left(\begin{matrix} 0 & \alpha \\ \beta & 0 \end{matrix}\right),\quad \alpha=\frac12(1+m),\quad \beta=\frac12(1-m),
%\]
\[
B_m := \frac{1}{2}(1+m)\begin{pmatrix} 0 & 1 \\ 0 & 0\end{pmatrix} + \frac{1}{2}(1-m)\begin{pmatrix} 0 & 0 \\ 1 & 0\end{pmatrix},
\]
with $m$ a constant mass term. We will also use $B = B_m$. We will focus on $m$ sufficiently close to $\pm 1$ as analyzing general $m \neq 0$ makes it difficult to prove the existence of a gap for $\Ht_n$. In more physical terms, we assume the laser is close to circularly polarized, and only allow slight ellipticity. We shall show that a gap opens at $E=0$
(See Figure \ref{fig:gaps} below).
%We may therefore define the projection of the Hamiltonian onto the negative energy component. The topology of that projection is then continuous with respect to changes to $m$ that maintain the gap. It is therefore enough to compute the invariants for $m=\pm1$ to obtain the invariants with slight ellipticity. In these cases, the matrix $B$ simplifies and then so do the calculations.
\begin{figure}[ht]
\begin{subfigure}{.45\textwidth}
\centering
\vspace{.5cm}
\includegraphics[width=.9\textwidth]{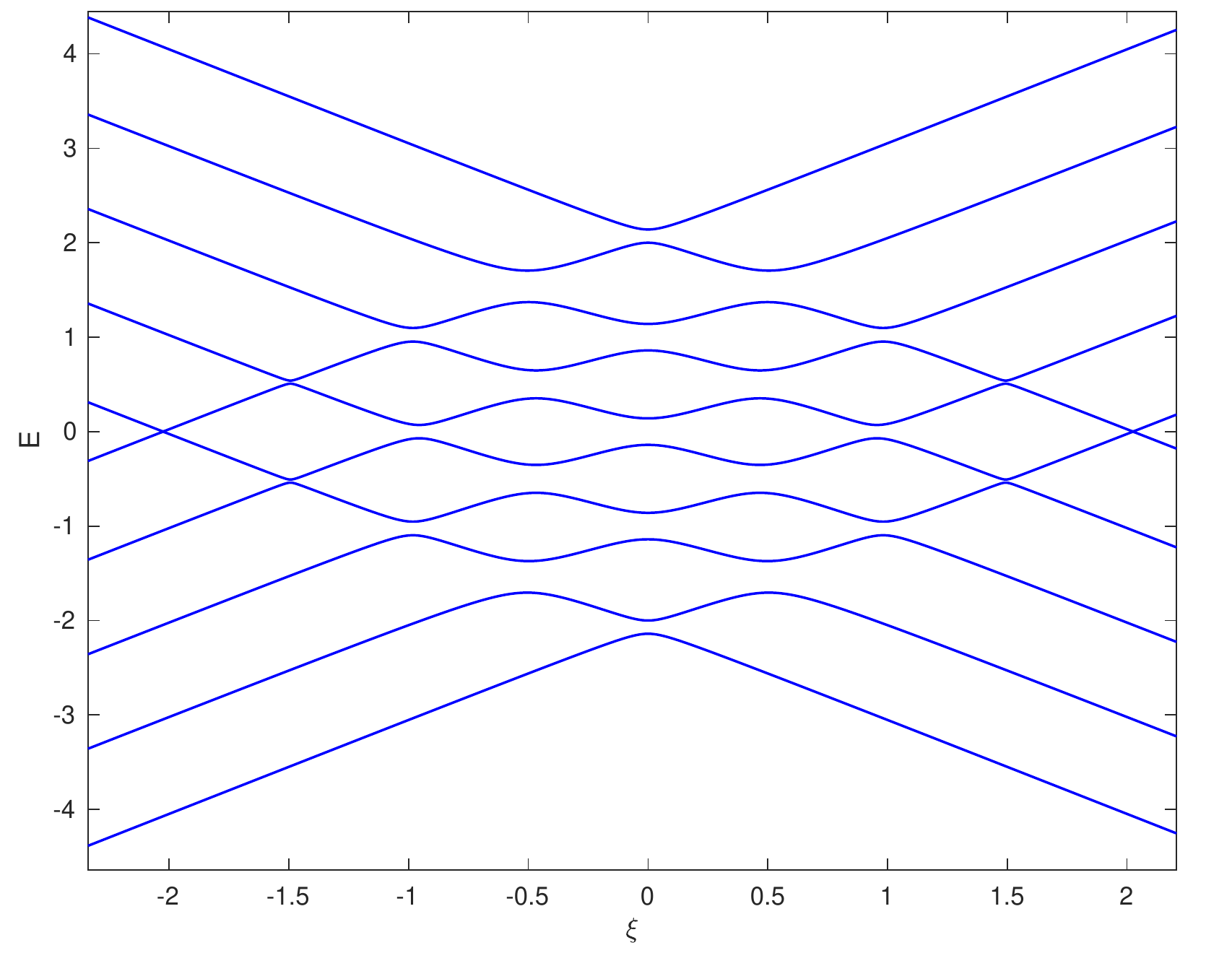}
\caption{Here we plot the bulk band structure of $\Ht_2$ on a line-cut through $\xi_2 = 0$. We see a sequence of shrinking gaps at $E=0$.}
\end{subfigure}
\begin{subfigure}{.45\textwidth}
\centering
\includegraphics[width=.9\textwidth]{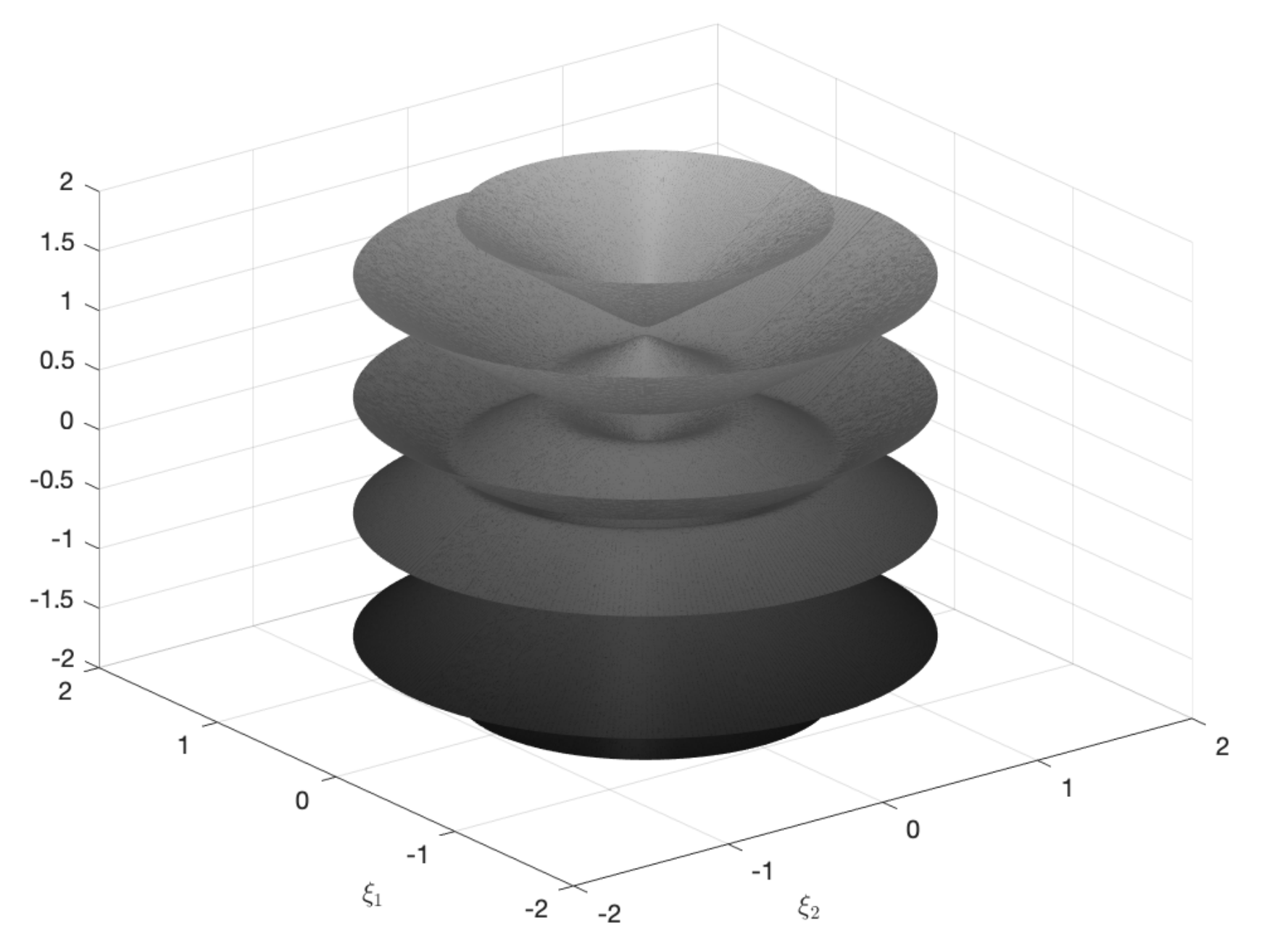}
\caption{Band structure of $\Ht_1$.}
\end{subfigure}
\caption{Cross-section for $\Ht_2$ and two-dimensional band structure for $\Ht_1$.}
\label{fig:gaps}
\end{figure}

\medskip
As a consequence, we will be able to define a bulk-difference invariant. 
 To do this, we define the eigenpairs of $\Ht_n(\xi)$ by $(h^{(i)},\psi^{(i)})$, $1\leq i\leq 2(2n+1)$. We have $2n+1$ branches $h^{(i)}>0$ for $i > 2n+1$ and $h^{(i)}<0$ for $i \leq 2n+1$. The calculation of the invariant under the assumption there is a gap in the vicinity of $E=0$ is given by  \cite[Equ. (23)]{bal3}
\begin{equation}\label{eq:WnTn}
  W_n =  \frac{i}{8\pi^2} \int T_n(\xi) d^2\xi
\end{equation}
with
\begin{equation}\label{eq:T}
  T_n(\xi)=4\pi i \sum_{i<j} \frac{\sgn{h^{(i)}}-\sgn{h^{(j)}}}{(h^{(i)}-h^{(j)})^2} \Im(\aver{\psi^{(i)}, \partial_1 \Ht_n(\xi) \psi^{(j)}}\aver{\psi^{(j)}, \partial_2 \Ht_n(\xi) \psi^{(i)}} ).
\end{equation}
This Kubo-type formula is arguably one of the simplest to use in the computation of the invariant.
We will however also need to use the form
\begin{equation*}
    W_n = i \sum_{j > 2n+1}\int_{\Rm^2} d(\psi^{(j)},d\psi^{(j)}),
\end{equation*}
see \cite{Hughes_2013,fruchart_2013}.

This invariant is not guaranteed to be integer valued as it is defined with an integration over a non-compact cycle $\Rm^2$. One way to remedy this situation is to construct a bulk-difference invariant \cite{bal3}.

We consider a value $m_0$ near $1$, and denote $W_n^+ = W_n$ for $m = m_0$ while we denote $W_n^- = W_n$ when $m = -m_0$. We now follow the gluing procedure in \cite{bal3}, to which we refer for details, to define the bulk-difference invariant. The projections onto the negative spectrum $\Pi_{\pm}(\xi)$ of the corresponding Hamiltonians $\Ht_n(\xi)$ with mass term $m = \pm m_0$ are easily seen to be independent of $\eps$ and $m$ as $|\xi|\to\infty$. Gluing the two planes $\xi\in\Rm^2$, each projected on half of a two-dimensional sphere $\Sm^2$, along the circle at infinity, we thus obtain a projector defined on a compact manifold (a sphere). This shows that 
\begin{equation*}
    W_{n}^d = W_n^- - W_n^+
\end{equation*}
is well-defined as an integer-valued bulk-difference invariant on the sphere. Since $\Ht_n$'s gap does not change as $m$ varies near $\pm 1$ (the existence of which has yet to be shown), the invariant does not change as $m_0$ varies continuously near $1$ and $\eps$ varies continuously in $(0,\eps_0)$ for $\eps_0$ sufficiently small.
\begin{thm}
 For $m_0$ sufficiently close to $1$ and $0<\eps$ sufficiently small, the bulk-difference invariant is well defined (i.e. there is a bulk gap at $E=0$), and is given by
\begin{equation}
W_n^d = 1 -2n(n+1).
\end{equation}
\end{thm}
\begin{proof}

At $\eps=0$, we observe for the unperturbed problem that 
\[
\aver{\psi^{(i)}, \partial_k \Ht_n(\xi) \psi^{(j)}}
\]
are purely imaginary so that the product of two such terms is purely real and the imaginary part in \eqref{eq:T} vanishes. The terms involving the imaginary parts are therefore vanishingly small with $\eps$. The only way to get a non-vanishing invariant is therefore when $h^{(i)}-h^{(j)}$ is small (with one positive and the other one negative). 
This occurs only for the two closest eigenvalue sheets about $0$, and only when $|\xi|$ is close to $\ell$ for $0\leq \ell\leq n$ (see Fig. \ref{fig:gaps}). Since the invariant does not change under continuous deformations that leave the gap open, we will consider the limit $\eps \rightarrow 0$ to compute the invariants. In particular, each ring $\ell$ (or ball when $\ell = 0$) in the limit have a bulk-difference invariant $\omega_\ell$ of their own that contributes to the total bulk-difference invariant $W_n^d$. We will show that $\omega_0 = 1$, $\omega_\ell = -4\ell$ when $l\geq1$, and
\[
W_n^d = \sum_{\ell = 0}^n \omega_\ell.
\]

We break the proof into three steps, which we outline as follows:
\begin{enumerate}
    \item For each ring (or ball) $0 \leq \ell \leq n$, we find a $2\times 2$ Hamiltonian $H_{2 \times 2}$ that controls the spectra to leading order in the gap size near $|\xi| \sim \ell$. When $\ell \neq 0$, we first introduce a $4\times 4$ Hamiltonian $H_{4\times 4}$.
    \item For the $m$ range of interest,  we verify a gap opening which scales as $\eps^{2\ell}$ near $|\xi| \sim \ell \neq 0$, and as $\eps^2$ near $|\xi| \sim 0$.
    \item We then compute the bulk-difference invariant for $H_{2\times 2}$ in the small $\eps$ limit $\omega_\ell$, and show it corresponds to the  contribution to $W_n^d$ from the ring in the small $\eps$ limit, $W_n^d = \sum_{\ell = 0}^n \omega_\ell$. For $\ell \neq 0$, this is done via $H_{4\times 4}$.
\end{enumerate}
We will focus primarily on the computations for $\ell \neq 0$ as they are the most intricate, and mention briefly at the end how to compute the contribution $\ell = 0$.
We also emphasize again that the term {\em replica} refers to the degrees of freedom of the vector space corresponding to a specific Fourier mode $\ell \in S_n$. In particular, a single replica $\ell$ corresponds to vectors $e_\ell \otimes \psi_\ell \in \Cm^{S_n \otimes \{1,2\}}$, $e_\ell$ a standard basis vector in $\Cm^{S_n}$ and $\psi_\ell \in \Cm^2$.

\medskip 
{\bf Step 1:} Constructing $H_{2\times 2}$ and $H_{4\times 4}$.

We consider $|\xi| \sim \ell \neq 0$. When building our leading order Hamiltonian approximations, we use that the spectrum is symmetric across $E=0$ and eigenvalues come in pairs $\pm E$. To verify this, we let
$\Theta$ be the matrix mapping $\Cm^{S_n}$ to itself defined by $\Theta_{ij} = \delta_{i+j}$. Then we let $\Theta' = \Theta\otimes \sigma_2$. We can then compute
  \[
  \Theta'\Ht_n(\xi)\Theta' = -\Ht_n(\bar\xi).
  \]
  However, these two operators have the same spectrum by simple conjugation so that the spectrum is symmetric across $E=0$.

We next wish to verify that there is a gap locally and we will verify the gap scales as $\varepsilon^{2\ell}$. To do this, we consider the eigenvalues $\pm E$ near $0$. When perturbation is turned off, the eigenstates correspond to replica $\ell$ and $-\ell$. When the perturbation is turned on, we will see the gap is opened by coupling between these two replicas via interaction through all the replicas in between (i.e. $-\ell+1,\cdots \ell -1$). 
To build $H_{4\times 4}$, we construct a leading order coupling of the $\pm \ell$ replicas. We use the following expression:
\begin{align}
\label{e:expand}
   & \psi_k = \varepsilon(k - E + \xi \cdot \sigma)^{-1} (\delta_{k < n}B \psi_{k+1} + \delta_{k > -n}B^*\psi_{k-1}),& |k| \neq \ell
\end{align}
to find the leading order coupling between blocks $\pm \ell$:
\[
  \psi_{\ell-1} = -\eps^{2\ell-1} (\ell-1-E+\xi\sigma)^{-1}B^*\ldots (-\ell+1-E+\xi\sigma)^{-1}B^* \psi_{-\ell} + O(\varepsilon^{2\ell}).
\]
Using a similar equation for $\psi_{1-\ell}$, we apply \eqref{e:expand} until only $\psi_\ell$ and $\psi_{-\ell}$ terms remain, and get the approximate eigenvalue problem up to $O(\eps^{2\ell+1})$ given by
\begin{align*}
   & E\psi_{-\ell,\ell} = H_{4\times 4}\psi_{\ell,-\ell} \\
   &H_{4\times 4} = \left(\begin{matrix}H_{\ell,\ell} & H_{-\ell,\ell}^*   \\ H_{-\ell,\ell} & H_{-\ell,-\ell}\end{matrix} \right) \\
   &H_{\pm \ell,\pm\ell} = \pm\ell+\xi\sigma +  K_{\pm\ell} \\
   &H_{-\ell,\ell} =-\eps^{2\ell}\bigl(\prod_{k=1-\ell}^{\ell-1}B(k+\xi\sigma)^{-1}\bigr)B.
\end{align*}
Here $K_{\pm\ell} = K_{\pm\ell}(\xi)$ include all terms in the expansion of order $\varepsilon^j$ for $1 \leq j \leq 2\ell$. $H_{-\ell,\ell}$ is the leading order coupling term between replicas $-\ell$ and $\ell$. We keep in mind that the point of this Hamiltonian is to approximate the two eigenvalues nearest $0$; the other two are not of interest for computing invariants as their contribution to $T_n(\xi)$ is negligible. As a consequence, we build a second Hamiltonian by projecting onto the two eigenstates with energies near $0$.
%We thus obtain a reduced approximate $4\times4$ system, which to leading order is given by
%\begin{equation*}
%\dm{H_{4\times4}} = \begin{pmatrix} \ell + \xi\sigma + \varepsilon D_1(\xi) &  \varepsilon(\varepsilon %\hat\xi_{-m})^{2\ell-1}cB_{-m} \\
%      \varepsilon(\varepsilon \hat\xi_m)^{2\ell-1} cB_m & -\ell + \xi\sigma + \varepsilon D_2(\xi) \end{pmatrix}.
%\end{equation*}
 Let 
\begin{equation}\label{eq:rotationxi}
    \tilde v_1 = \frac{1}{\sqrt{2}}\begin{pmatrix} \overline{\hat{\xi}} \\ 1\end{pmatrix}, \quad \tilde v_2 = \frac{1}{\sqrt{2}}\begin{pmatrix} -\overline{\hat{\xi}} \\ 1 \\ \end{pmatrix},\quad v_1= \begin{pmatrix} 0 \\ 1\end{pmatrix} \otimes \tilde v_1, \ \  \mbox{ and } \ \  v_2 = \begin{pmatrix} 1 \\ 0 \end{pmatrix} \otimes \tilde v_2.
\end{equation}
%We let $v_1= \begin{pmatrix} 0 \\ 1\end{pmatrix} \otimes \tilde v_1$ and $v_2 = \begin{pmatrix} 1 \\ 0 \end{pmatrix} \otimes \tilde v_2$.
Then our reduced matrix to leading orders using $\eta = \eps^{-2\ell}\tilde v_1^*H_{-\ell,\ell}\tilde v_2$ is given by
\begin{equation*}
\begin{split}
    H_{2\times 2} &=  \{ v_i^*H_{4\times 4}v_j\}_{ij} 
    %\\&
    = \begin{pmatrix} \ell - |\xi| +  \kappa(\xi) & \eps^{2\ell}\bar\eta \\ \eps^{2\ell}\eta & -(\ell - |\xi|) - \kappa(\xi) \end{pmatrix}
    = (\ell-|\xi|+\kappa(\xi))\sigma_3  + \eps^{2\ell}\eta\cdot \sigma.
\end{split}
\end{equation*}
Here, $\kappa(\xi) = \tilde v_1^*K_{\ell}(\xi)\tilde v_1 = -\tilde v_2^*K_{-\ell}(\xi)\tilde v_2$. The latter values are the same by symmetry of the spectrum about $E=0$. We observe here that the term $\kappa(\xi)\sigma_3$ cannot open a gap without the $\sigma_1,\sigma_2$ terms as $(\ell-|\xi|)\sigma_3$ moves across the gap in the first diagonal entry and the second, and dominates the influence of $\kappa(\xi)$, which scales as $O(\eps)$. We  note that these terms shift the minimal gap location perturbatively away from $|\xi| = \ell$. 

{\bf Step 2:} Verifying local gap near $|\xi| \sim \ell \neq 0$ scaling as $O(\eps^{2\ell})$. 
% Radial Symmetry
%We first verify radial symmetry of $\Ht_n$'s spectrum in $\xi$ for $m \in \pm 1$. This is easy to verify by considering $\det(H_n(\xi) - \lambda)$. When expanding the determinant, every time one chooses a term of the form $\bar\xi$ when performing row eliminations, one forces a corresponding $\xi$ from the same replica to be included thus leaving the determinant phase-independent. This follows since $B = e_1\otimes e_2$ or $B = e_2\otimes e_1$, which has a convenient kernel for row expansions.
We show $\eta \neq 0$ for $m=\pm 1$, which by perturbation theory is sufficient to show that there is a local gap scaling as $O(\eps^{2\ell})$ near the ring with radius $\ell$ for $m$ sufficiently close to $\pm 1$. 
Since the result holds for all values of $l$, this shows a gap opening for $m$ sufficiently close to $\pm1$ and $0<\eps\leq\eps_0$ sufficiently small.
Simultaneously we will explicitly calculate $\eta$ for use in Step 3.
%
%
% Verify gap exists for general m
%

 \begin{prop}
   For any wavenumber $\xi$ with $|\xi|$ close to $l$ and $m = \pm 1$, we obtain that 
   $$\eta = -mc(\xi)\hat\xi^{2m\ell }\qquad \mbox{ with } \qquad 
   c(\xi) = \biggl(\prod_{k=-\ell+1}^{\ell-1}\frac{|\xi|}{|\xi|^2-k^2}\biggr).$$
 \end{prop}
 \begin{proof}
We recall and define the following notation to keep track of our choice of $m$:
\begin{align*}
    B_1 = \begin{pmatrix} 0 & 1 \\ 0 & 0\end{pmatrix}, \quad B_{-1} = \begin{pmatrix} 0 & 0 \\ 1 & 0 \end{pmatrix}, \quad
    &\Lambda_1 = \begin{pmatrix} 0 & 0 \\ 0 & 1 \end{pmatrix}, \quad \Lambda_{-1} = \begin{pmatrix} 1 & 0 \\  0 & 0\end{pmatrix}.
\end{align*}
Observe that for $m \in \{\pm1\}$, $B_m\Lambda_{m} = B_m$, $\Lambda_{m}B_m = 0$, and $(B_m)^2 = 0$.
We also observe 
\begin{align*}
    &\xi \cdot \sigma B_m = \xi_m \Lambda_m 
    &  \xi \cdot \sigma \Lambda_m = \bar\xi_mB_m%\xi \cdot \sigma B_m = B_{-m} \bar\xi\cdot\sigma
\end{align*}
where $\xi_{1} = \xi$ and $\xi_{-1} = \bar \xi$. Then for $\gamma$ any scalar
\begin{equation*}
        (\gamma+\xi\cdot\sigma) \frac{|\xi|}{\gamma^2-|\xi|^2} \biggl(\frac{\gamma}{|\xi|}B_m - \hat\xi_m\Lambda_m\biggr) =\frac{1}{\gamma^2-|\xi|^2}\biggl( \gamma^2B_m + \gamma \xi_m \Lambda_m - \gamma \xi_m\Lambda_m - |\xi|^2 B_m\biggr),
\end{equation*}
and hence
\[
(\gamma+ \xi\cdot\sigma)^{-1}B_m = \frac{-|\xi|}{|\xi|^2-k^2}\biggl(\frac{\gamma}{|\xi|} B_m - \hat\xi_m \Lambda_m\biggr).
\]
Now we observe
\begin{equation}
\label{e:B_Lambda}
        B_m(\frac{\gamma}{|\xi|} B_m - \hat\xi_m \Lambda_m) =-\hat\xi_m B_m.
\end{equation} 
This gives us:
\begin{equation}
\begin{split}
    B_m\prod_{k=-\ell+1}^{\ell-1}B_m(k +\xi\cdot\sigma)^{-1}B_m &= cB_m\prod_{k=-\ell+1}^{\ell-1}\biggl( \frac{k}{|\xi|}B_m - \hat\xi_m \Lambda_m\biggr)
    = (-\hat\xi_m)^{2\ell-1}cB_m.
\end{split}
\end{equation}
Applying $\tilde v_1^*$ to the left and $\tilde v_2$ on the right picks up an additional phase and sign $m\hat\xi_m$, and the proposition is complete.
\end{proof}

{\bf Step 3:} Computing bulk-difference invariants of $H_{2\times 2}$ in the small $\eps$ limit and estimate their contribution to the true bulk-difference invariant.

 Henceforth we consider only $m = \pm 1$ as the integer-valued invariants for $m$ near $\pm 1$ are identical to those of $m = \pm 1$ by continuity. 
 
 Since we now know that $W_n^d$ is integer-valued as we demonstrated the presence of a spectral gap (technically, we have not shown this for $|\xi|$ close to $0$ yet; this is done below) and we also showed that it was independent of $0<\eps<\eps_0$, any contribution to \eqref{eq:WnTn} that is small as $\eps\to0$ may therefore be safely ignored.  By construction, the invariant of $H_{4\times4}$ is therefore asymptotically the only contribution to the invariant of interest in the integral \eqref{eq:WnTn} for values of $|\xi|$ close to $l$. 
 
 Continuous deformations show that we may remove $K_{\pm \ell}$ from $H_{4\times 4}$ and $\kappa$ from $H_{2\times 2}$ in the invariant computations. We therefore replace the $2\times 2$ and $4\times 4$ Hamiltonians respectively with
\begin{align*}
    &H_{2\times 2} = (\ell - |\xi|)\sigma_3 -mc \varepsilon^{2\ell} \hat\xi^{2m\ell}\cdot\sigma \\
    &H_{4\times 4} = \begin{pmatrix} \ell + \xi\sigma  &  \varepsilon(\varepsilon \hat\xi_{-m})^{2\ell-1}cB_{m}^* \\
      \varepsilon(\varepsilon \hat\xi_m)^{2\ell-1} cB_m & -\ell + \xi\sigma \end{pmatrix}.
\end{align*}
We now need to elucidate one point. The reduction from the $4\times4$ to the $2\times2$ Hamiltonians in \eqref{eq:rotationxi} depends on $\xi$ and it is therefore not clear a priori that the invariant for the $4\times4$ systems can be computed using the invariant for the $2\times2$ system. We need to verify the invariant for $H_{2\times 2}$ is the same as the invariant of $H_{4\times 4}$ in the small $\eps$ limit. To do so, we use the computation of invariants using the connections rather than the curvatures knowing that these two computations are related by a simple application of the Stoke's theorem (since curvature is defined as the exterior derivative of the connection) \cite{Hughes_2013,fruchart_2013}.

Defining $V = (v_1,v_2)$ and $(\varphi^\pm,\pm E)$ as the eigenpairs of $H_{2\times2}$, we find to leading order (in $\eps$) that
\begin{equation*}
    H_{4\times4} V\varphi^{\pm} \approx \pm E V\varphi^{\pm}.
\end{equation*}
In the computation of the bulk-difference invariant for $H_{4\times4}$ using an integral such as \eqref{eq:WnTn}, we can approximate the contributions for $|\xi|$ close to $\ell$ using $m = -1$ by
\begin{equation*}
\begin{split}
    \omega_{\ell}& = 2i\int_{\mathcal{C}_\ell} d(V\varphi^+,d (V\varphi^+)) 
    = 2i \int_{\mathcal{C}_\ell} d(\varphi^+,d\varphi^+) + d(\varphi^+, (V^*dV)\varphi^+)
\end{split}
\end{equation*}
where $C_\ell = \{ z \in \Cm: |z| \in [\ell-1/2,\ell+1/2]\}$. 
Now we wish to show the second term is $0$. That the second term above vanishes shows that the invariants, which may be computed as line integrals of connections instead of volume integrals of curvatures \cite{Hughes_2013,fruchart_2013}, are indeed identical. We find:
$
    V^*dV = \hat\xi d\widehat{\bar\xi}.
$
Hence
\begin{equation*}
    \int_{\mathcal{C}_\ell} d(\varphi^+,(V^*dV)\varphi^+) = \int_{\partial \mathcal{C}_\ell} \hat\xi d\overline{\hat\xi} = 0.
\end{equation*}
The final integral is zero as we integrate over two oppositely directed circles.

It finally remains to compute the invariant of the limiting $2\times2$ system.
We use the coordinates $\xi=e^{i\theta}(n+\eps^{2n} r)$ for $(r,\theta)\in\Rm\times(0,2\pi)$.
 The system becomes after dividing by $\eps^{2\ell}$, and to leading order $|\xi|=1+\eps^{2\ell}r$,
\[
    \left( \begin{matrix}
     -r & \alpha e^{-2\ell i\theta} \\ \alpha e^{2\ell i\theta}& r
\end{matrix} \right) \varphi   =  (-r\sigma_3 +\alpha \cos(2\ell\theta)\sigma_1+\alpha \sin (-2\ell\theta)\sigma_2)\varphi = E \varphi,
\]
for some constant $\alpha\not=0$. Consider the slightly more general family of Hamiltonians
\[
  H_p = \cos(p\theta+\phi) \sigma_1 + \sin(p\theta+\phi)\sigma_2 + \tau r \sigma_3.
\]
We are interested in the case $p = -2\ell$ and $\tau = -1$.
We find
\[
  \partial_1 H=\tau\sigma_3,\quad \partial_2 H =-p \sin (p\theta+\phi) \sigma_1 + p \cos (p\theta+\phi) \sigma_2 \] and \[
[\partial_1 H,\partial_2 H] = -2i \tau p(\sin(p\theta+\phi) \sigma_2 + \cos(p\theta+\phi) \sigma_1) 
\]
so that 
\[
  \int_0^{2\pi} {\rm tr} H [\partial_1 H,\partial_2 H] d\theta = -8\pi i p \tau.
\]
This result is independent of the phase $\phi$. Note that $H^2=(1+r^2)I$. We thus compute half the contribution to the ring's bulk-difference invariant $-\frac{p}{2}\sgn{\tau}$ when integrating
\[
   \frac i{2\pi} \dint \dfrac{-1}{8|H|^3} {\rm tr} H[\partial_1 H,\partial_2 H] dk.
\]
Since $\tau=-1$ above and $p = -2\ell$, we obtain
\[
\omega_\ell = -4\ell
\]
with the extra factor of $2$ coming from the fact this is a bulk-difference invariant (which we recall means computing the difference of the above integrals evaluated for $m=\pm1$). This concludes the computation of the contributions to the invariant coming from $|\xi|\sim l$ for $l\geq1$.

\medskip
{\bf The case $|\xi|$ small.}
It thus remains to compute the contribution to the invariant coming from $\xi\sim 0$. We do it when $n=1$ to slightly simplify notation. The generalization to arbitrary $n$ follows similar machinery from the previous case. We want to eliminate non-contributing terms and write an equation for the middle component of the system
\[
  \left(\begin{matrix} 1-E+\xi\cdot\sigma & \eps B_m^* & 0 \\ \eps B_m & \xi\cdot\sigma-E & \eps B_m^* \\ 0 & \eps B_m & -1-E+\xi\cdot\sigma \end{matrix} \right) \left(\begin{matrix}\psi_1 \\ \psi_0 \\ \psi_{-1}\end{matrix} \right)=0.
\]
For $\xi$ and $E$ close to $0$, all diagonal terms but the middle one are invertible. We obtain the result
\[
  \left( -  \left(\begin{matrix} \eps B_m^* \\ \eps B_m \end{matrix} \right)^*  \left(\begin{matrix} 1-E+\xi\cdot\sigma & 0\\ 0 & -1-E+\xi\cdot\sigma \end{matrix} \right)^{-1}  \left(\begin{matrix} \eps B_m^* \\ \eps B_m \end{matrix} \right) + \xi\cdot\sigma-E\right) \psi_0 =0.
\]
This opens a gap, so that we now officially know that $H_n$ has a spectral gap near $E=0$ for $\eps$ small enough since the only possible remaining obstruction was for $|\xi|$ small. The gap close to $\xi=0$ is well approximated by the system
\[
   (\xi\cdot\sigma - E -m\eps^2 \sigma_3 )\psi_0 =0.
\]
The invariant for such an operator is $-\text{sign}(m)\frac12$ so that the bulk-difference invariant equals $\omega_0 = 1$.

We thus obtain the two reduced systems of interest when two eigenvalues are close to the gap $E=0$. 
We find a contribution to the bulk-difference invariant equal to $\omega_0 = 1$ for the contribution close to $|\xi|=0$ and equal to $\omega_\ell = -4\ell $ for the contribution close to  $|\xi|=\ell$.  The topology of the projection onto the negative part of the energy of $H_n$ is thus given by the winding number (at energy $E=0$)
\[
 W_n^d = \sum_{\ell = 0}^n \omega_\ell = 1-4\sum_{\ell=1}^n \ell = 1-2n(n+1).
\]
For $n=1$ (the $3-$replica model) for example, we find $W_1=-3$.
\end{proof}

\subsection{Interface conductivity}
\label{sec:cond}

Let $n$ be fixed and $H=\Ht_n$ one of the above replica models or $H=\htwo$ the reduced $2\times2$ system defined in \eqref{eq:htwo}. Bulk-difference invariants (when $m$ is constant) were computed for these Hamiltonians in the preceding section. We already know from the discussion in the introduction (or indeed from the calculations in the preceding section) that the bulk-difference invariant for $\htwo$ is $W_0^d=1$.

The above bulk computations translate to a quantization of an interface conductivity $\sigma_I$ in \eqref{eq:sigmaI} when the now spatially varying mass term $m(y)$ has different signs as $|y|\to\infty$. We assume that $m(y)$ is smooth and is equal to $m_0>0$ for $y>y_0>0$ and equal to $-m_0$ for $y<-y_0$, where $m_0$ is sufficiently close to $1$ as discussed in the previous section so that gaps open for these values. This translates into a gap $|E|>E_0$ for the bulk Hamiltonian with $m=\pm m_0$. We use semiclassical calculus results derived in \cite{bal3} to show that there indeed is a well defined interface current, which is quantized according to the bulk invariants.

Many results on the correspondence between interface and bulk invariants are available in the literature both for time-dependent (Floquet) topological insulators or not; see for instance \cite{elbau2002equality,graf2018bulk,prodan2016bulk,sadel2017topological,volovik2003universe}.

We can also write 
$$H = D \cdot \gamma_0 + \eps \gamma(y)$$
for $\gamma_0 = \sigma \otimes I_n$, $\sigma = (\sigma_1,\sigma_2)$. For $y < -y_0$, $\gamma = \gamma_-$ and $y > y_0$, $\gamma = \gamma_+$. Here we use 
$$\gamma = \eps^{-1} (H-D \cdot \gamma_0).$$ 
Hence $\gamma(y)$ has constant coefficients away from the interface, and $\gamma(y)$ is smooth across the interface. We know the corresponding constant coefficient Hamiltonians have bulk gaps as shown in the previous section. This is the structure of PDEs considered in \cite{bal3}, which we will apply here to establish the bulk-interface correspondence for $H$. 

We define $P$ as a smooth non-decreasing function of $x \in \Rm$ such that $P(x)=0$ for $x\leq -x_0<0$ and $P(x)=1$ for $x\geq x_0$. We let $\varphi$ be a smooth function with $\varphi(u) = 0$ for $u < -E_1$ and $\varphi(u) = 1$ for $u > E_1$, where $0 < E_1 < E_0$ and $[-E_0,E_0]$ is within the bulk gap. We then define the edge conductivity:
\[
\sigma_I = \Tr\ i[H,P]\varphi'(H).
\] The main result of this section will then be the bulk-interface correspondence:
\begin{thm}
For the system defined as above and $\eps$ sufficiently small, we have the bulk-interface correspondence
\[
2\pi\sigma_I = -W_n^d.
\]
\end{thm}
\begin{proof}
The proof relies on verifying the conditions of Proposition 4.7 (quantization of $\sigma_I$) and Corollary 4.15 (bulk-interface correspondence) of \cite{bal3}. Most of the conditions are trivial, but we do need to verify conditions on our operator done through operator Weyl symbols. A matrix valued operator $A$ with values in $\mathbb{M}_{2(2n+1)}$ ($2(2n+1)\times2(2n+1)$ matrices) can be represented in terms of its Weyl symbol $a(x,\xi) \in \mathcal{S}'(\Rm^2\times\Rm^2; \mathbb{M}_{2(2n+1)})$:
\begin{align*}
    &A = \Weyl(a) ,\qquad 
    \Weyl(a)\psi(x) = \frac{1}{(2\pi )^2}\int_{\Rm^2\times\Rm^2} e^{i(x-y)\cdot \xi}a(\frac{x+y}{2},\xi)\psi(y)dyd\xi.
\end{align*}
To apply Proposition 4.7 of \cite{bal3} to obtain quantization of $\sigma_I$, we first need to show $(I+H^2)^{-1}$ has Weyl symbol with appropriate decay. Let $a$ be the symbol of $I+H^2$. We must show
$$(I + H^2)^{-1} = \Weyl(\tilde a)$$
for some $\tilde a \in S^{-2}$. Here $S^m$ is defined as the Fréchet space of functions  satisfying
\begin{equation}
\label{e:symbol_decay}
|\partial_{(\alpha,\beta)}a(x,\xi)| \leq C_{\alpha,\beta} \langle \xi \rangle^{m-\beta}.
\end{equation}
To verify the form of $(I+H^2)^{-1}$, we refer to the argument from Equation (8.10) in \cite{sjostrand}  and the application of Beals' criteria (Proposition $8.3$ in \cite{sjostrand}) to verify $\tilde a \in S^0$. This also provides us an operator $R = \Weyl(r)$ with $r \in S^{-1}$ such that
\[
A^{-1} = \Weyl(\tilde a) = \Weyl(a^{-1}) - \Weyl(\tilde a)R.
\]
$R$ is thus a smoothing operator, and the parametrix $\Weyl(a^{-1})$ dominates the decay estimates. Hence $\tilde a \in S^{-2}$ since $a^{-1} \in S^{-2}$.

The conditions (h1)-(h2) for Corollary 4.15 of \cite{bal3} are trivially verifiable using the symbol of $H$, $\xi \cdot \gamma_0 + \gamma(y)$. Thus we obtain the bulk-interface correspondence, which concludes the proof.
\end{proof}

In other words, we find that $2\pi\sigma_I(\Ht_n)=-1+2n(n+1)$. This is confirmed numerically for $n=1$ as well as for the central gap corresponding to $n=0$  (see Figure \ref{fig:edge}). For $n = 0$, recall that $\htwo = D\cdot\sigma - \eps^2m(y)\sigma_3$ after a bit of algebra, hence yielding a standard gapped Dirac edge state.
\begin{figure}[ht]
\centering
\begin{subfigure}{.45\textwidth}
\vspace{.9cm}
\centering
\includegraphics[width=1\linewidth]{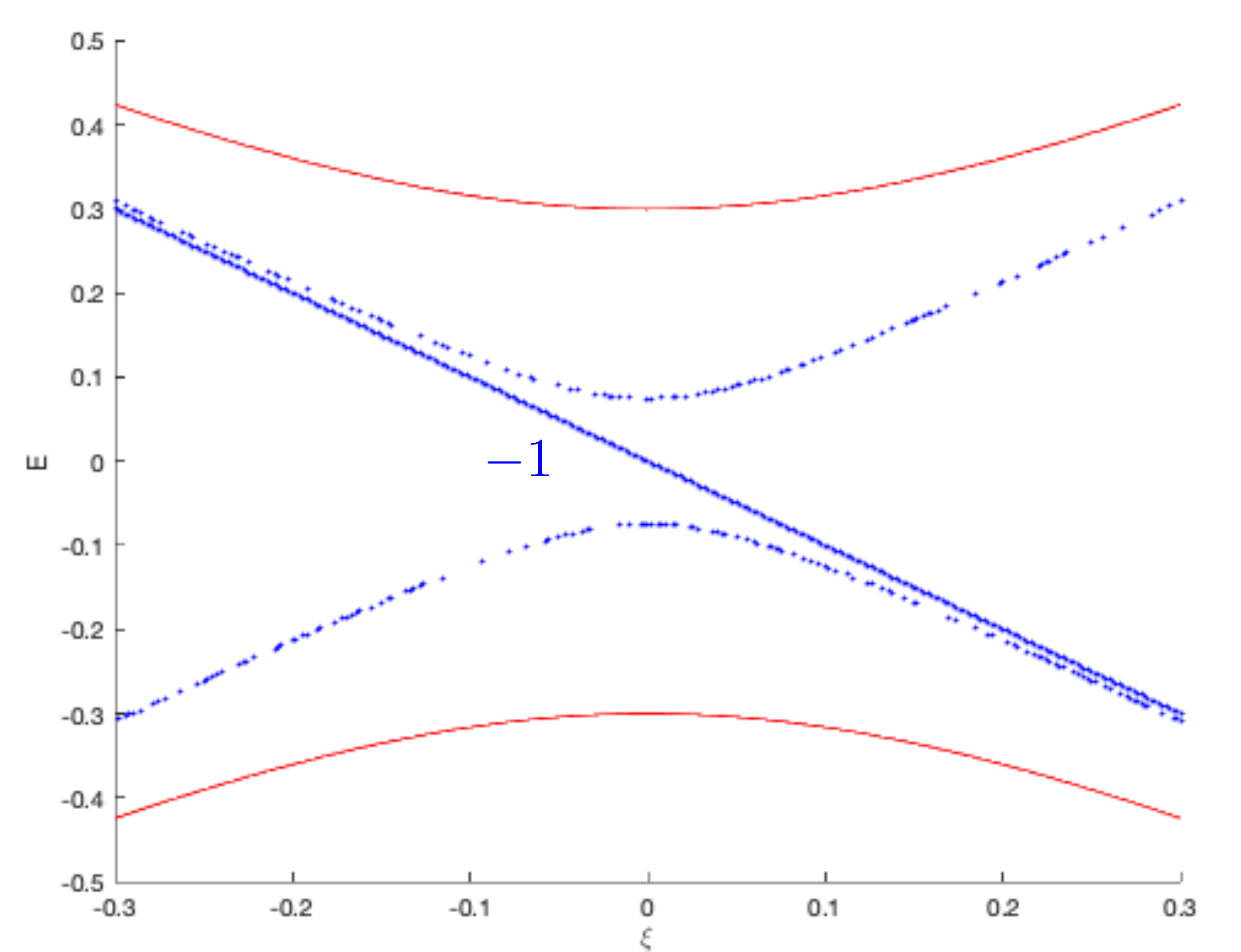}
\caption{Here we plot the edge state for $\htwo$ with $\Omega = 10/3$ in rescaled units. Two edges were included, and we plot only edge states localized to one side of the two-edge system.}
\end{subfigure}
\hspace{2mm}
\begin{subfigure}{.45\textwidth}
\centering
\includegraphics[width=1\linewidth]{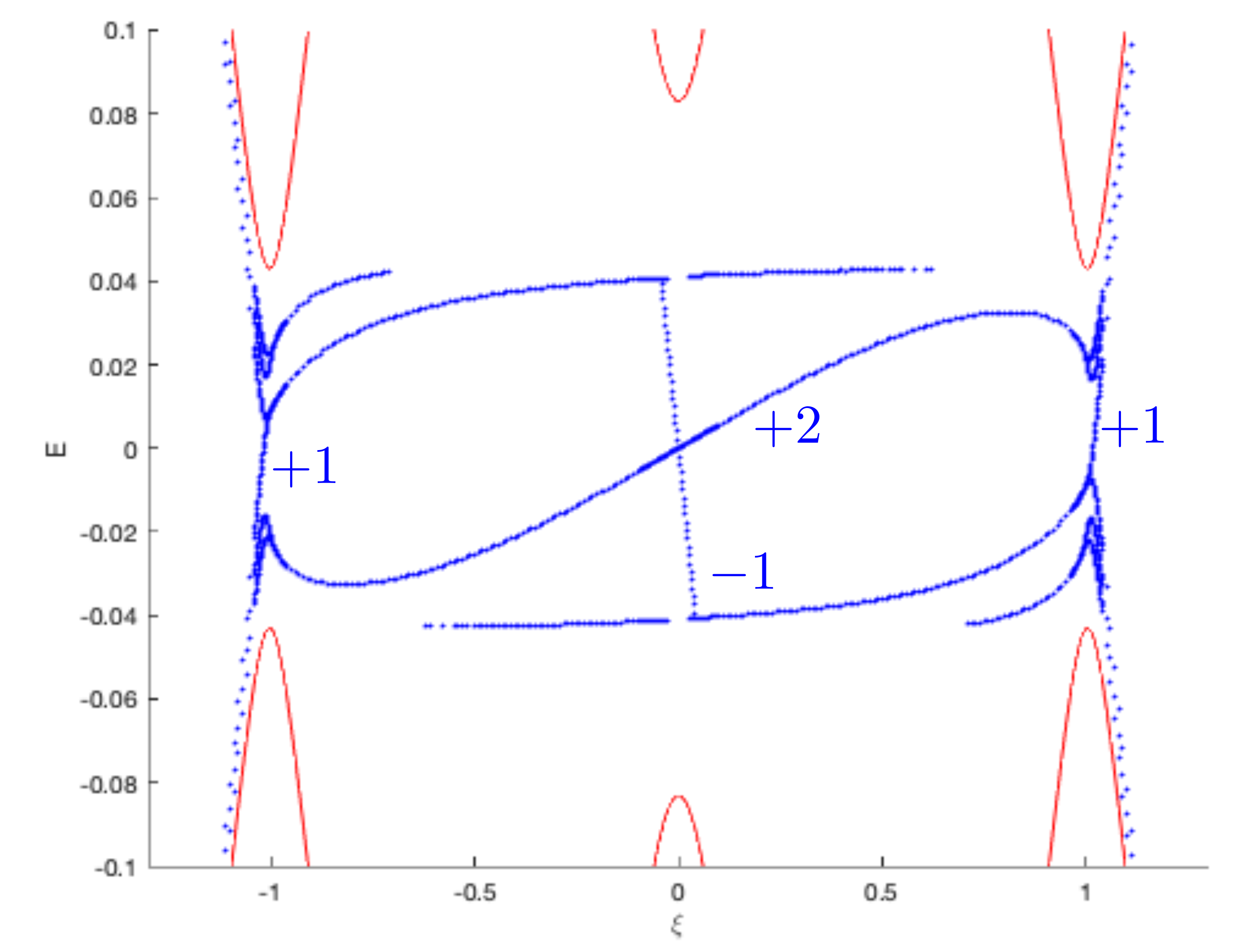}
\caption{Edge states for $\Ht_1$ for $\Omega = 10/3$ in rescaled units. }
\end{subfigure}
\caption{Replica model edge states plotted against bulk quasi-band structure.}
\label{fig:edge}
\end{figure}

We interpret these results as follows. For sufficiently small times (less than $\Omega=\eps^{-1}$ for a coupling $B$ of order $O(1)$ in the original variables), then the topology of the central gap is given by $n=0$ with a mode propagating from right to left along the $x-$axis with a topology given by $W_0^d=1$.

For longer times, heuristically less than $\Omega^{3}=\eps^{-3}$, then interactions between the propagating modes in the $n=1$ model (one mode propagating to the left as above with 4 modes propagating to the right) become significant. For sufficiently long times, an equilibrium takes place with a topology given by $W_1^d=-3$. For yet longer times, the modes of the $n=2$ model (the $5-$replica model) also participate in the transport (provided that $\varphi'$ is now so concentrated that it is supported in the spectral gap of the $5-$replica model). We then expect transport to respect a topology given by $W_2^d=-11$. 

The topology $W_0^d=1$ is likely to be of most interest practically. It should govern interface transport for times that are long but not too long (at most $\Omega=\eps^{-1}$); see also \cite{mciver2020light} for large pulsed irradiations that open sizeable gaps without damaging the underlying material.
For longer times, many physical phenomena (e.g., heat generation) not included in $\Ht_n$ may become more prevalent. If $\Ht_n$ remains a valid model, then we should observe a transition from the $W_0^d$ `topology' to the $W_1^d$ 'topology' in the presence of sufficient scattering among all propagating modes; see discussion in concluding section. %What type of scattering may result in such equilibria is analyzed for a simple model in the first paper.
\section{Averaging theory}
\label{sec:at}

This section considers a class of Hamiltonians for which explicit effective Hamiltonians may be derived by an averaging method in the high frequency regime.
%, as in the preceding sections. 
Consider the unperturbed Hamiltonian $H_0=D\cdot\sigma$. A general electromagnetic time-dependent perturbation takes the form
\begin{equation}
\label{eq:peiels}
  H(t) = (D+A(t))\cdot\sigma + V(t)
\end{equation}
with $A(t)$ and $V(t)$ the magnetic and electric potentials, respectively. As in the preceding sections, the Hamiltonian $H(t)$ has no component along the gap-opening component $\sigma_3$.

The objective is to propose large, high-frequency modulations $(A,V)$ such that a non-trivial topology emerges for an effective Hamiltonian. As in the preceding section, the effective Hamiltonian dominates dynamics only for times that are not-too-large (small compared to the driving frequency, assumed to be large).

We still denote by $\Omega$ the driving frequency and by $\eps=\frac1{t_0\Omega}\ll1$ for $t_0\equiv1$ a reference time scale. We then construct the Hamiltonian
\begin{equation}\label{eq:Heps}
    H_\eps(t) = D\cdot\sigma + \frac1\eps \Big(f_1(\frac t\eps) \sigma_1 + f_0(\frac t\eps)v(y)\Big).
\end{equation}
This corresponds to the choice $A_1(t)=\frac1\eps f_1(\frac t\eps)$, $A_2(t)=0$, and $V(t)=\frac1\eps f_0(\frac t\eps)v(y)$. The functions $f_j(\tau)$ are chosen to be $t_0-$periodic so that $f_j(\frac t\eps)$ is $\Omega^{-1}-$periodic.

The factor $\frac1\eps$ reflects the fact that the rapid fluctuations need to be large to have an order $O(1)$ effect as $\eps\to0$. The term $f_1(t)$, a spatially constant rapidly oscillating magnetic potential creates the necessary twisting to acquire a non-trivial topology. The potential term $f_0(t)v(y)$ also requires an appropriate time evolution as we shall see while the spatial component $v(y)$ provides the edge confinement as the effective bulk topology depends on the sign of $v$. 

The last important ingredient in the above structure is that the fast Hamiltonian
\[
   H_{-1}(\tau) = f_1(\tau) \sigma_1 + f_0(\tau) v(y)
\]
is explicitly integrable since the commutator $[H_{-1}(\tau_1),H_{-1}(\tau_2)]=0$ for any times $\tau_1,\tau_2$. Indeed, let $F_j(\tau)$ be the antiderivatives of $f_j(\tau)$, i.e., $F_j'(\tau)=f_j(\tau)$ with $F_j(0)=F_j(t_0)=0$ assuming $\int_0^{t_0}f_j(\tau)d\tau=0$, as we do for the rest of the section.

The evolution associated to $H_{-1}(\tau)$ is given by the unitary
\[
  U_{-1}(\tau) = e^{-i \int_0^\tau H_{-1}(s)ds} = e^{ -i [F_1(\tau)\sigma_1 + F_0(\tau) v(y)]} = e^{-iF_0(\tau)v(y)}\big( \cos(F_1(\tau)) I -i \sin(F_1(\tau)\sigma_1 \big),
\]
where we used that 
\[
  e^{ia\sigma_1} =  \cos a\ I + i \sin a\  \sigma_1.
\]
Let $U_\eps(t)$ be the unitary evolution associated with $H_\eps(t)$, i.e., the solution of
\[
  i \partial_t U_\eps(t) = H_\eps(t) U_\eps(t)
\]
with $U_\eps(0)=I$. Factoring out the fast evolution, we may introduce
\[
   \tilde U_\eps(t) = U_{-1}^*(\frac t\eps) U_\eps(t)
\]
and obtain that 
\begin{equation}\label{eq:tildeH}
    i\partial_t \tilde U_\eps(t) = \tilde H_\eps(t) \tilde U_\eps(t),\qquad \tilde H_\eps(t) = U_{-1}^*(\frac t\eps) H_0 U_{-1}(\frac t\eps).
\end{equation}
Introducing $\tilde H(\tau)=\tilde H_\eps(\eps \tau)$ and $\tilde U(\tau)=\tilde U_\eps(\eps\tau)$, we obtain using shorthand $F_j = F_j(\tau)$
\begin{equation*}
    \begin{split}
        \tilde H(\tau) &= U_{-1}^*(\tau)H_0 U_{-1}(\tau) \\
        &= e^{iF_0v(y)}(\cos F_1 I + i \sin F_1 \sigma_1) D\cdot\sigma (\cos F_1 I - i \sin F_1\sigma_1) e^{-i F_0v(y)} \\
        &= \frac{1}{i}\partial_x \sigma_1 + (\cos F_1 I + i \sin F_1 \sigma_1)\sigma_2 (\cos F_1 I - i \sin F_1\sigma_1) (-F_0v'(y) + \frac{1}{i}\partial_y) \\
        &= \frac{1}{i} \partial_x \sigma_1 + (\cos F_1 I + i\sin F_1\sigma_1)^2 \sigma_2(-F_0v'(y) + \frac{1}{i}\partial_y) \\
        &= (\cos 2F_1 \sigma_2 - \sin 2F_1\sigma_3)(-F_0v'(y) + \frac{1}{i}\partial_y).
    \end{split}
\end{equation*}
 Thus we have
\begin{equation}\label{eq:tH}
    \tilde H(\tau) = \Big( \cos(2F_1(\tau))\sigma_2 - \sin(2F_1(\tau))\sigma_3\Big) \Big( \frac1i\partial_y - F_0(\tau)v'(y)\Big) +\frac1i \partial_x \sigma_1.
\end{equation}
Here, we used that $\sigma_1\sigma_2=i\sigma_3$. 
We now observe at the fast scale that
\[
   i\partial_\tau \tilde U (\tau) = \eps \tilde H(\tau) \tilde U(\tau).
\]
Since the influence is small, it is reasonable to expect that the main effect of $\tilde H$ is felt through its time average, at least up to moderately large times where additional effects may appear. To prove this, we use a two-scale averaging framework in a functional setting adapted to the differential operator $H_0$.

Let us introduce 
\begin{equation}\label{eq:aH}
    \aver{H} = \frac{1}{t_0}\int_0^{t_0} H(\tau) d\tau.
\end{equation}
\begin{prop}\label{prop:aver}
  Let $H(t)$ be a $t_0-$periodic Hamiltonian with a scale of Hilbert spaces $\mH_s$ such that $H(t)$ is (uniformly in time) bounded from $\mH_s$ to $\mH_{s+1}$ for $s=0,1$ and generates a unitary evolution in $\mH_0$. 
  Let $\aver{H}$ be the time averaged operator and $\varphi\in\mH_2$ be a sufficiently smooth initial condition. 
  
  Consider the evolutions 
  \[
     i\partial_t \psi_\eps = H(\frac t\eps) \psi_\eps,\qquad i\partial_t \psi = \aver{H}\psi
  \]
  both with initial conditions $\psi_\eps(0)=\psi(0)=\varphi$.
  
  Then we have the approximation
  \[
    \|\psi_\eps(t) - \psi(t) \|_{\mH_0} \leq C t\eps \sup_{0\leq s\leq t}\|\psi(s)\|_{\mH_2}
  \]
  for a constant $C=C(t_0,H)$ independent of $\eps$ and $t$.
\end{prop}
This shows that the averaging approximation holds up to times that are small compared to $\frac1\eps t_0$ times the supremum  $\sup_{0\leq s\leq t}\|\psi(s)\|_{\mH_2}$. This depends on the structure of the averaged Hamiltonian. For times of order $O(t_0)$, such a supremum is bounded by standard regularity theory so that the error in the quantum dynamics is proportional to $\eps$, or equivalently $\Omega^{-1}$. For longer times, this supremum may grow polynomially for some initial conditions $\psi$. We do not consider this well-studied problem and refer instead to \cite{maspero2017time} and its references for details.

\begin{proof}
  In a two-scale formalism, we replace the $\eps-$dependent problem formally up to $O(\eps^2)$  terms in the expansion of $\psi$ by
\[
  (\frac i\eps\partial_\tau + i\partial_t )(\psi_0+\eps\psi_1)= H(\tau)(\psi_0 + \eps\psi_1).
\]
Here we assume $\psi_j(t,\tau)$ is periodic in $\tau$.
Solving these equations in turns gives $\psi_0=\psi_0(t)$ and then 
\[
  i\partial_t \psi_0 = \aver{H} \psi_0 ,\qquad i\psi_1(t,\tau) = \Big(\dint_0^\tau (H(s)-\aver{H}) ds\Big)\psi_0(t) + \psi_{10}(t).
\]
We choose $\psi_{10}(t)=0$ and $\psi_0(0)=\varphi$ while $\psi_1(0)=0$. We can use these terms in a Hilbert expansion
\[
   i\partial_t (\psi_0(t)+\eps\psi_1(t,\frac t\eps) + \zeta_\eps ) = H(\frac t\eps) ( \psi_0(t) +\eps\psi_1(t,\frac t\eps) + \zeta_\eps)
\]
and obtain that
\[
  i\partial_t \zeta_\eps = H(\frac t\eps) \zeta_\eps + S_\eps
\]
where the source term is given by
\[
   S_\eps = \eps i \Big[ \dint_0^{\frac t\eps} (H(s)-\aver{H})ds \aver{H} - H(\frac t\eps)\dint_0^{\frac t\eps} (H(s)-\aver{H})ds  \Big] \psi_0.
\]
This is a term of order $O(\eps)$ provided that $\psi_0$ is sufficiently smooth. The integrals in time occur over an interval or size bounded by $t_0$ since $H$ is periodic. This provides an error estimate for $\zeta_\eps(t)$ as given in the proposition. The term $\eps\psi_1$ is bounded similarly.
%
%[The powers of $t_0$ are not consistent.]
\end{proof}
We now apply the above result to the operator in \eqref{eq:Heps}, for which we choose $\mH_s=H^s(\Rm^2)$ the standard Sobolev space of functions with $s$ derivatives in $L^2(\Rm^2)\equiv H^0(\Rm^2)$.
\begin{corollary}\label{cor:aver}
  Let $H_\eps(t)$ be the modulated operator given in \eqref{eq:Heps} and assume that $v'(y)$ is smooth and uniformly bounded on $\Rm$. The effective Hamiltonian is given by
  \begin{equation}\label{eq:averH}
      \tilde H = \frac 1i\partial_x \sigma_1 + Y \frac 1i\partial_y + M v'(y)
  \end{equation}
  with constant matrices $Y$ and $M$ given by
  \[
    Y = \aver{\cos (2F_1)}\sigma_2 + \aver{\sin(2F_1)}\sigma_3,\quad 
    M = -(\aver{\cos(2F_1)F_0}\sigma_2 + \aver{\sin(2F_1)F_0}\sigma_3).
  \]
  Let $\psi_\eps$ be the solution of 
  \[
    i\partial_t \psi_\eps = H_\eps \psi_\eps,\qquad \psi_\eps(0) = \varphi
  \]
  for $\varphi\in \mH_2$.
  Define $\psi$ as the solution to the effective evolution
  \[
     i\partial_t\psi = \tilde H\psi,\qquad \psi(0)=\varphi.
  \]
  Then we have 
  \[
    \|\psi_\eps(t) - U_{-1}(\frac t\eps) \psi(t)\|_{\mH_0} \leq C t\eps \sup_{0\leq s\leq t}\|\psi(s)\|_{\mH_2}.
  \]
\end{corollary}
\begin{proof}
This is a direct corollary of the preceding result and the fact that $U_{-1}$ is unitary on $\mH_0$. That $\sup_{0\leq s\leq t}\|\psi(s)\|_{\mH_2}$ is uniformly bounded at least for times of order $O(1)$ is a standard regularity result. In general, we expect such a supremum to grow as a function of $t$ although we do not consider the details here; see \cite{maspero2017time}.
\end{proof}

Let us consider the case with $F_1(\tau)$ and $F_0(\tau)$ odd with respect to $\frac12 t_0$. Then 
\[
  \aver{\sin(2F_1)} = \aver{\cos(2F_1)F_0)} =0
\]
by oddness. Defining
\[
  h_y = \aver{\cos(2F_1)} \quad \mbox{ and } \quad m(y) = -\aver{\sin(2F_1)F_0)} v'(y),
\]
we obtain the effective Hamiltonian
\[
  \tilde H = \frac1i \partial_x \sigma_1 + h_y \frac1i \partial_y \sigma_2 + m(y) \sigma_3
\]
which has a nontrivial topology when $h_y>0$, say, and $m(y)$ is bounded away from $0$ away from $y=0$ and has different signs as $y\to\pm\infty$.

More generally, let us define a matrix $B$ with entries
\[
 b_{11}= \aver{\cos(2F_1)}, \ b_{12}= -\aver{\cos(2F_1)F_0)},\ b_{21}= \aver{\sin(2F_1)},\ b_{22} = -\aver{\sin(2F_1)F_0)},
\]
and assume that $B$ has non-vanishing determinant. Let us also assume to simplify that $m'(y)$ is continuous and non-vanishing so that it has constant sign. Then we find that the interface conductivity (invariant) is given by 
\[
 2\pi \sigma_I = - {\rm sgn} ({\rm Det} B)\  {\rm sgn} (m'(0)). 
\]
These formulas are obtained as we did in earlier section. We leave the details to the reader.
%There are many ways to check this formula. This is a direct outcome of the calculations in a paper with Solomon.

\medskip

The above results show that the local density $|\psi_\eps|^2(t,x)$ of $\psi_\eps(t)$ is accurately described by that of $\psi(t)$, the solution of a topologically non-trivial Hamiltonian dynamics. Indeed, the unitary $U_{-1}(\tau)$ is locally unitary in the sense that $|U_{-1}(\tau)\psi|^2=|\psi|^2$ for any two-spinor $\psi$. However, such an approximation is a priori valid only over times that are short compared to the driving frequency $\Omega$ (or in proper units, $\Omega|H|^{-2}$ as may be inferred from the proof of the above proposition, with $|H|$ a frequency quantifying the strength of the Hamiltonian $H$).

%This is a consistent picture with the one obtained in the preceding section. 

\section{Conclusions}
\label{sec:conclu}

The topology of a material, as is the case for a manifold, concerns its global structure. It is immune to continuous deformations by construction and this makes topological invariants useful in practice when they can be associated with physical behaviors. In topological insulators, the edge conductivity \eqref{eq:sigmaI}, characterizing global properties of transport along an interface between insulators, is one such physically relevant invariant.

In a model such as \eqref{eq:H2x2}, the topology of the model is that of the vortex $\xi\cdot\sigma$ in momentum space, characterized by the winding of $\frac{\xi_1+i\xi_2}{|\xi|}$ around a `circle at infinity'. It is this behavior at infinity, combined with the behavior at infinity of the mass term $m(y)$ that characterizes the quantized values of $\sigma_I$ in \eqref{eq:valsig}. While such infinite domains are unrealistic but convenient in many modelings and applications, here they are central to the definition of the topology and a modeling choice we make.

Physically, $\sigma_I$ non vanishing indicates that transport along the edge has to occur. However, it does not fully describe edge transport. Let us assume that in a given energy range within the bulk band-gap, $m$ modes are allowed to propagate along the $x$ axis in the positive direction and $n$ modes in the negative one. Then $2\pi\sigma_I=m-n$ in \eqref{eq:valsig}. Neither $m$ nor $n$ are topologically protected separately. In the presence of minimal coupling among the modes, then $m$ modes will propagate rightward and $n$ modes leftward independently of what $\sigma_I$ indicates. However, in the presence of strong coupling (or equivalently over long times), then (Anderson) localization effects prevent ${\rm min}(m,n)$ from propagating, and asymptotically in the strong coupling regime, only $m-n$ modes propagate; see \cite[Theorem 6.2]{B-EdgeStates-2018} in the context of \eqref{eq:H2x2}, where it is shown that propagation across a highly heterogeneous slab asymptotically results in $m-n$ transmitting modes and ${\rm min}(m,n)$ totally reflected modes. Topology alone cannot protect against backscattering. It rather protects against the localization of exactly $2\pi\sigma_I=m-n$ modes. 

This paper analyzes how much of the above picture remains valid in the context of time-periodically driven materials such as graphene. As we mentioned in the introduction, opening a gap (with $m(y)$ large enough so that $\varphi'(H)$ has enough bandwidth) is a non-trivial task. Provided that the electromagnetic drive does not heat the material too rapidly, the Hamiltonian description of the light-matter interaction in \eqref{eq:EMH} is reasonable. In such a context, which needs to be posed on an infinite domain $\Rm^2$ if we want to model the presence of (non-periodic) impurities (since $\sigma_I$ is independent of the presence of a large class of impurities \cite{bal2}), the Hamiltonian $H(t)$ is locally (in time) trivial in the sense that it does not display any spectral gap, and the associated unitary evolution $U(t)$ also does not display any spectral gap at any time. It therefore seems difficult to define 
%the logarithm of $U(T)$ and 
an explicit effective Hamiltonian as well as any topology based on a three dimensional winding number \cite{rudner2013}.
%%% REMARK: We can always define U(T) spectrally as \int_0^{2\pi} e^{i\lambda d\pi_\lambda for an appropriate spectral measure and then define H=\int_0^{2\pi} \lambda d\pi_\lambda with obviously U(T)=e^{iH}. But this operator is not explicitly defined. Not sure we need to write this in detail.

Instead, what we show is that emergent topologically non-trivial effective Hamiltonians appear at different time scales for times that are long compared to the forcing period $T$. In the absence of scale separation, the coupling between the replica levels (see Fig.\ref{fig:gaps}(a)) may be arbitrarily complicated and results in complex transport patterns; see \cite[Chapter 5]{fruchart:tel-01398614} for relevant numerical simulations. The scale separation $\Omega\gg1$ allows for the perturbation expansion considered in sections 2 and 4. We showed that the unitary evolution $U(t)$ was well approximated by $U_n(t)$ based on the $n-$replica model up to times of order $\Omega^{n+1}$. We also showed an approximation of $U(t)$ by an evolution $U_0(t)$ valid up to times of order $\Omega^{2}$ provided the initial condition is sufficiently smooth. Associated to these evolutions $U_n$ are a sequence of gapped effective Hamiltonians $H_n$. Even though $U_n$ converges to $U$ in operator norm over increasing ranges of time, the interface conductivities $\sigma_I(H_n)$ diverge as $n$ increases and no $\sigma_I$ can be assigned to the time-dependent Hamiltonian $H(t)$.

As a result, we obtain a sequence of effective conductivities $2\pi\sigma_I(H_0)=-1$, $2\pi\sigma_I(H_1)=3$, and so on, which reflects the topology that may be perceived at different time scales. Let us comment on the two panels in Fig. \ref{fig:edge}. Let us consider a density of states given by $\varphi'(H_0)$ concentrated in the central gap of the left panel in Fig. \ref{fig:edge}. For times of order $O(1)$, these wave packets evolve according to that dispersion relation with one more mode moving left-ward (with a negative group velocity) than right-ward. As time increases and in the presence of coupling among modes in that energy range (as shown e.g. in \cite{B-EdgeStates-2018} for the Dirac operator \eqref{eq:H2x2}), the other modes that are compatible energetically become populated until an equilibrium given by the conductivity associated to $H_1$ is reached, with $2\pi\sigma_I(H_1)=3$ more modes moving right-ward than left-ward. Higher-order conductivities may then be relevant for even longer time scales (of order $\eps^{-n+1}$ for the $n-$replica model) provided that the considered wavepackets live in the $n$th bulk gap (of width $\eps^{2n}$ for $n\geq1$). Heuristically, the system undergoes a cascade of phase transitions from $2\pi\sigma_I(H_j)$ to $2\pi\sigma_I(H_{j+1})$ for appropriately confined wavepackets until the Hamiltonian description no longer holds because of, e.g, dissipative effects.

In the last section, we considered a different model where a direct averaged effective Hamiltonian may be explicitly computed. In the high-frequency, high-amplitude regime, we show at least theoretically that large gaps may be opened in graphene over time scales that are (i) not-too-long compared to $\Omega$ and (ii) small enough that the dissipation-less Hamiltonian models remains relevant. The main feature that allowed us to obtain an effective Hamiltonian explicitly was to use a fast driving force that involves commuting Hamiltonians (i.e., not circularly polarized light). We stress that, here as in the preceding sections, the topology is associated to an effective Hamiltonian, whose validity is demonstrated only over not-too-long times.

\section*{Acknowledgments}

We would like to thank Michel Fruchart for enlightening discussions on the topic of FTIs. This research was partially supported by the National Science Foundation, Grants DMS-1908736 and EFMA-1641100 and by the Office of Naval Research, Grant N00014-17-1-2096.

\bibliographystyle{abbrv}
\bibliography{fti}

\end{document}